\newif\ifpublic
\publictrue
\ifpublic
\documentclass[10pt,a4]{article}
\else
\documentclass[11pt,draft,letterpaper]{article}
\fi
\PassOptionsToPackage{hyphens}{url}

\usepackage[margin=1in]{geometry}
\usepackage[final,pagebackref,colorlinks=true,linkcolor=blue,urlcolor=blue,citecolor=blue,pdfstartview=FitH,breaklinks=true]{hyperref}
\usepackage[T1]{fontenc}    
\usepackage{amsmath}       
\usepackage{amssymb}
\usepackage{amsthm}
\usepackage{amsbsy}
\usepackage{nicefrac}       
\usepackage{lipsum}
\usepackage{fullpage}
\usepackage{mathpazo}
\usepackage{mathtools}
\usepackage{enumitem}
\usepackage{subfiles}
\usepackage{xcolor}
\usepackage{setspace}
\usepackage{algorithm}%
\usepackage{algorithmic}%
\usepackage{graphicx}
\usepackage{textcomp}
\usepackage{pifont}
\ifpublic
\usepackage[disable]{todonotes}
\else
\usepackage[colorinlistoftodos]{todonotes}
\fi
\newcommand{\cmark}{\ding{51}}%
\newcommand{\xmark}{\ding{55}}%

\newcommand{\specialcell}[2][c]{%
  \begin{tabular}[#1]{@{}c@{}}#2\end{tabular}}%

\newcommand*\samethanks[1][\value{footnote}]{\footnotemark[#1]}

\usepackage[capitalize,nameinlink]{cleveref}
\crefname{algocf}{alg.}{algs.}
\Crefname{algocf}{Algorithm}{Algorithms}
\renewcommand{\eqref}[1]{\hyperref[#1]{(\ref*{#1})}}

\usepackage{pgfplots}
\usepackage{appendix}
\pgfplotsset{compat=1.14}
\theoremstyle{plain}
\newtheorem{theorem}{Theorem}[section]

\newtheorem{lemma}[theorem]{Lemma}

\newtheorem{definition}[theorem]{Definition}

\theoremstyle{definition}
\newtheorem{remark}[theorem]{Remark}

\renewcommand{\epsilon}{\varepsilon}

\renewcommand{\phi}{\varphi}
\newcommand{\polylog}{\mathrm{polylog}}

\newcommand{\pg}[1]{\noindent\textbf{#1}~~}

\title{OverChain: Building a robust overlay with a blockchain}
\author{
  Vijeth Aradhya\thanks{National University of Singapore, Singapore. email: {\tt
      \{varadhya,seth.gilbert,hobor\}@comp.nus.edu.sg}.}
   \and
     Seth Gilbert\samethanks
   \and
 Aquinas Hobor\samethanks
}

\begin{document}
\maketitle
\tableofcontents
\pagebreak
\begin{abstract}
Blockchains use peer-to-peer networks for disseminating information among peers, but these networks currently do not have any provable guarantees for desirable properties such as Byzantine fault tolerance, good connectivity and small diameter. This is not just a theoretical problem, as recent works have exploited unsafe peer connection policies and weak network synchronization to mount partitioning attacks on Bitcoin. Cryptocurrency blockchains are safety critical systems, so we need principled algorithms to maintain their networks.
 
Our key insight is that we can leverage the blockchain itself to share information among the peers, and thus simplify the network maintenance process. Given that the peers have restricted computational resources, and at most a constant fraction of them are Byzantine, we provide communication-efficient protocols to maintain a hypercubic network for blockchains, where peers can join and leave over time. Interestingly, we discover that our design can \emph{recover} from substantial adversarial failures. Moreover, these properties hold despite significant churn.
 
A key contribution is a secure mechanism for joining the network that uses the blockchain to help new peers to contact existing peers. Furthermore, by examining how peers join the network, i.e., the ``bootstrapping service,'' we give a lower bound showing that (within log factors) our network tolerates the maximum churn rate possible. In fact, we can give a lower bound on churn for any fully distributed service that requires connectivity.
 \end{abstract}

\section{Introduction}

Blockchains, distributed ledgers, and most other distributed services rely on an ``overlay network'' to facilitate communication among the users of the service. For example, cryptocurrencies like Bitcoin \cite{nakamoto2008peer} rely on the peer-to-peer network to provide timely and efficient communication among the peers.  The primary goal of this paper is to give a new Byzantine-resilient algorithm for maintaining an overlay designed for blockchain networks with high rates of churn, proving it correct, robust, and efficient.  

The overlay network is specifically designed to integrate with a blockchain; in fact, the key insight in this paper is that by leveraging the blockchain itself to share critical information, we can develop a protocol that is simpler and more efficient than existing solutions. (In some ways, this is a ``cross-layer optimization,'' with the overlay protocol relying on the blockchain for which it provides the underlying communication.) For example, by using the blockchain, the overlay protocol does not have to pay the cost of running its own consensus protocols, and does not have to rely on complicated Byzantine-resilient structures. Instead, it can focus on the simpler task of maintaining an efficient communication network. 


\pg{Obstacles for overlay maintenance.} Three main challenges arise while building and maintaining an overlay network designed for use by a (public) blockchain: high rates of churn, the problem of introducing new peers, and malicious behavior.

When peers join and leave the network, the overlay needs to adapt, integrating the new peers and removing the departing peers, while maintaining the good properties of the network, e.g., small diameter and small degree. Worse, joins and leaves do not necessarily happen one at a time: many peers can (concurrently) join and leave the network at any given instant. This poses a unique challenge for peers to continually reconstruct the network in an efficient way.  

In fact, how new peers securely join the network is itself a critical aspect of a peer-to-peer system, commonly referred to as the ``bootstrapping problem''~\cite{castro2002one, conrad2007generic, dickey2008bootstrapping}.  To join, new peers need to contact some existing (honest) peers, and since overlay networks are decentralized, it is often not clear as to who should be responsible for helping new peers. Currently, cryptocurrency blockchains offer two main options for bootstrapping peers: (1) via DNS seeding and (2) hard-coded IP addresses (in the shipped software) \cite{loe2019you}. Hard-coded seeds are usually a fallback mechanism if the DNS seeding mechanism fails, as hard-coded IP addresses may become obsolete after a period of time. DNS seeding itself requires trusting a few sources to respond with random IP addresses that are within the network. Moreover, if there is a bandwidth constraint on each peer or DNS server, and there are only a handful of DNS servers responsible for bootstrapping, then the system cannot withstand high churn.

The challenge of overlay maintenance is only made worse by malicious behavior.  Malicious peers can create a large number of identities, allowing them to attack an overlay in a myriad of ways. They might attack the underlying maintenance protocol, conveying bad information about network structure and the peers joining or leaving. Or they might attack the overlay itself by strategically being concentrated in one part of the network, possibly with an aim to partition the network (or isolate a single peer), for e.g., as in an \textit{eclipse attack} \cite{singh2004defending, singh2006eclipse}.  

In an eclipse attack, the attackers isolate a subset of peers, and by taking over the connections of those peers, the adversary is able to control the information flow between them and the rest of the network. In the context of cryptocurrency blockchains, this attack can be used as a primitive for other attacks such as double spending, reducing \textit{effective} honest resources (by forcing a significant fraction of peers to mine on top of an \textit{orphan} block), and selfish mining \cite{nayak2016stubborn, eyal2014majority}.

Early eclipse attacks \cite{heilman2015eclipse, cryptoeprint:2018:236} exploited unsafe peer storage and connection policies. The adversary set up many incoming connections with the victim peer from diverse IP addresses and propagated many bogus addresses. Then, there was a high likelihood that the victim would initiate all its outgoing connections with the adversary once the victim restarted. Recently, Saad et al. \cite{saad2021revisiting} exploited Bitcoin's weak synchronization to create a partition by selectively broadcasting blocks to disjoint groups of miners. In a parallel work, Saad et al. \cite{saad2021syncattack} exploit churn, and partition existing and newly arriving peers. The adversary occupies all the incoming connections of existing peers, and as peers depart and new peers join, they only connect with adversarial peers from the sample provided by the DNS seeds.




\pg{Prior work.}  The traditional approach for providing resilience against Byzantine interference in the well-studied distributed hash table (DHT) literature is \emph{replication} \cite{naor2007novel, fiat2007censorship, awerbuch2009towards}. Thus, our starting point is a virtual network, specifically a hypercube, in which each vertex of the hypercube is implemented in a replicated fashion by a set of peers. We will refer to the set of peers that collectively replicate a vertex as a \emph{committee}. Such a replication method is useful for ensuring that there are a sufficient number of honest peers in each committee.

However, replication alone is not enough to guarantee robustness in networks where peers can join and leave over time \cite{awerbuch2004group}. For example, a subset of honest peers can stay for a long time, while other honest peers experience churn, and the Byzantine peers repeatedly rejoin to isolate those honest peers. The join algorithm becomes crucial in such a setting as it is not only used to maintain the overall network structure but also in ensuring that the malicious peers are well-spread throughout the committees. Even if there exists a join algorithm that places a new peer in a random committee, Byzanine peers can overwhelm a committee with just a linear number of rejoins \cite{scheideler2005spread}. Therefore, join algorithms are typically complemented by ``perturbing'' the network, where a small number of (existing) peers are shuffled among some committees \cite{scheideler2005spread, awerbuch2009towards, guerraoui2013highly}.


Unfortunately, existing solutions to such join-leave attacks are fairly expensive, both from a latency and message complexity perspective (cf. Table \ref{tab:comparison}). For instance, consider integrating a new peer and placing it in a random committee. The system would have to run heavy-weight protocols such as distributed coin flipping \cite{awerbuch2004group, fiat2005making, awerbuch2009towards}, or rely on random walks \cite{guerraoui2013highly}, to sample a random committee. Or if the new peer locally uses a hash function to compute a random committee, then the network needs to continually generate a (global) random string as a (partial) input to avoid pre-computation attacks \cite{jaiyeola2018tiny}. Furthermore, the new peer also needs to be routed to the appropriate committee. Such protocols involve a logarithmic join latency\footnote{As in \cite{jaiyeola2018tiny}, we employ a Sybil defense mechanism to control the number of malicious identities. Other prior works \cite{awerbuch2004group, fiat2005making, awerbuch2009towards, guerraoui2013highly} are also vulnerable to Sybil attacks but they assume such a Sybil mechanism already exists. Join latency does not include the time taken to solve the puzzle given for Sybil defense.}, and are hard to implement in practice.

Moreover, existing works do not address the bootstrapping problem, which is crucial for the join algorithms. They make strong assumptions such as the existence of trusted peers that can initiate an unlimited number of joins, or access to random peers within the network. Jaiyeola et al. \cite{jaiyeola2018tiny} point out that secure bootstrapping would aid their algorithms and other prior solutions to robust overlay maintenance. Our goal is to include the bootstrapping process as in integral part of the overlay protocol.

As with any long-lived system, it is possible eventually for something to go wrong. Consider exceptional scenarios, termed as \emph{catastrophic failures} (cf. Section \ref{sec:recovery}), where malicious peers instantly overwhelm a large number of committees. The existing solutions \cite{awerbuch2004group, fiat2005making, awerbuch2009towards, guerraoui2013highly, jaiyeola2018tiny} heavily rely on honest majority in committees for the correctness of their algorithms, thus making it hard to recover from such a scenario. One of the interesting properties of our design is that the overlay recovers fairly naturally from such disasters.

\begin{table}
\begin{center}
\begin{tabular}{ |c|c|c|c|c| }
\hline
        & \specialcell{Join latency\\(in rounds)} & \specialcell{Join comm.\\complexity} & \specialcell{Polynomial\\variation in\\network size} & \specialcell{Recovery\\(catastrophic\\failures)}\\ 
\hline
 Group spreading \cite{awerbuch2004group} & $O(\log N)$ & $O(\log^3 N)$ & \xmark & \xmark\\
\hline 
 S-Chord \cite{fiat2005making} & $O(\log N)$ & $O(\log^3 N)$ & \xmark & \xmark\\
\hline 
 Cuckoo rule \cite{awerbuch2009towards} & $O(\log N)$ & $O(\log^3 N)$ & \xmark & \xmark\\
\hline
 NOW \cite{guerraoui2013highly} & $O(\log^4 N)$ & $O(\log^6 N)$ & \cmark & \xmark\\
 \hline
 OverChain & $O(1)$ & $O(\log^3 N)$ & \cmark & \cmark\\
 \hline
\end{tabular}
\caption{Comparison under different performance metrics.}
\label{tab:comparison}
\end{center}
\end{table}

\subsection*{Our approach} In general, overlay maintenance amidst churn and Byzantine interference requires considerable coordination among peers. Our insight is that this coordination issue is exactly what blockchains are designed to solve! Thus, we store auxiliary data on the blockchain to efficiently maintain the overlay. The idea of exploiting on-chain information to simplify and facilitate off-chain distributed algorithms already exists in practice. For example, payment channel networks\footnote{They (and most other Layer-2 solutions) make implicit assumptions regarding blockchain not forking (safety) and the maximum time taken to confirm a transaction on the blockchain (liveness).} such as Bitcoin's Lightning Network \cite{poon2016bitcoin}, enable fast transaction confirmation if users are willing to lock funds in them.

An existing model that captures something of the same idea is a ``public bulletin board'' \cite{mitzenmacher2000useful} where entities within the system can write information that can be read by everyone. A cryptocurrency blockchain differs from a typical bulletin board in three ways: (1) the amount of auxiliary information in a block must be small, (2) the rate at which information can be shared is limited by the block interval, and (3) the network may never be fully synchronized where each peer holds the same chain. Thus, one of our contributions is to carefully distill the properties provided by the blockchain in a way that the overlay algorithms can be concisely described while not losing track of real world implementation.


An important question to ask, then, is how much information overlay algorithms need to store on the blockchain. We could, for example, try to store the  entire membership and topology information --- but we would then be spending a large majority of our blockchain bandwidth in handling the overlay! Instead, our goal is to use only a \emph{constant} number of entries per block. Specifically, each block stores the identity of only one peer in the overlay. To maintain a virtual hypercube with at most $N$ peers, we rely on a set of about $\Theta(\sqrt{N})$ of the most recent blocks which store a set of about $\Theta(\sqrt{N})$ peers, which we refer to collectively as a \emph{directory}. The fairness of the blockchain ensures that not too many of these peers will be malicious. Each peer in the directory is responsible for a subset of the committees, and keeps tracks of the members of those committees.

A crucial observation is that cryptocurrency blockchains are publicly available, and their contents can be read by anyone. Specifically, a recent copy of the blockchain is available to the public at all times, and this provides an entry point to the service. Blockchain explorers \cite{bitinfocharts, blockchaincom, etherchain} satisfy this requirement to some extent; multiple copies may be available but the \textit{confirmed}\footnote{In the blockchain literature, there exists a notion of a ``confirmed'' chain for which the sequence of blocks (up to that point) will not change in the future.} chain ensures some synchronization, i.e., there might be some disagreement about the last few blocks among those explorers.

Another problem we face is that malicious peers can create a large number of identities and take control of the overlay.  Again, the blockchain has already solved this problem, typically via proof-of-work.  We adopt the same mechanism for limiting the rate at which new peers join the network.

\pg{Overview.} In this way, most of the coordination required by the overlay is handled simply by storing a small amount of information on each block, leading to a significantly simpler approach than is typical for Byzantine-resilient overlay protocols. A new peer gets a public copy of the blockchain, and computes a proof-of-work puzzle that determines its (random) committee. The peer then contacts a small number of directory members to quickly obtain information about existing peers in that committee. The committees (or directory members) do not need to run a consensus protocol, perform random walks, or do any other sort of coordination. Figure \ref{fig:qe-img} gives a pictorial overview of our design. 

\begin{figure}[tb]
\centering
\includegraphics[width=0.5\linewidth]{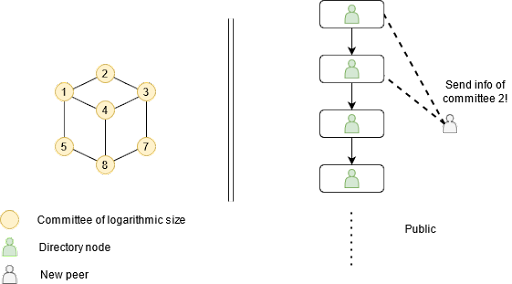}
\caption{A new node requesting information about existing nodes in a committee from a set of directory nodes.}
\label{fig:qe-img}
\end{figure}

\pg{Churn.} Furthermore, we seek to understand the limits of the \emph{rate} of churn. We adopt the \emph{half-life} approach to churn rate for honest peers: if there are $H$ peers in the system, then over a specific interval of time, the half-life, at most $H/2$ new peers can join or at most $H/2$ peers can depart. This allows for highly bursty behavior, with large numbers of concurrent joins and departures. (Malicious peers can stay for as long as they like, and can create new identities as fast as the proof-of-work mechanism will let them.) We provide a lower bound for the feasible half-life that depends on the rate of churn and the bandwidth constraints. 


As new blocks are installed on the blockchain, and as existing blocks age, the members of the directory change, handing off information from old directory members to new directory members in a controlled process. Similarly, when the number of peers changes significantly, the size of the virtual hypercube has to change, migrating information to new directory members. Both of these processes of information exchange have to be carefully managed to avoid Byzantine interference.


\pg{Recovery.} Finally, the blockchain is not just an alternative interface for new peers to join, it also aids the overlay to recover from catastrophic failures. As long as most of the committees and directory are still functioning properly, our observation is that the overlay operates sufficiently well to continue installing new blocks, to continue replacing directory and committee members, and restoring the fully correct operation of the overlay, i.e., to ensure that there are again sufficient number of honest peers in every committee.

\subsection*{Summary of results} We exploit the blockchain for bootstrapping peers, (global) coordination among peers (for e.g., agreement on new topology, etc.) and recovery from arbitrary committee failures. We summarize our contributions (that hold with high probability\footnote{In this paper, ``with high probability'' (abbreviated as ``whp'') refers to a probability of at least $1-N^{-k}$ where $N$ is the number of peers and $k > 1$ is an appropriate constant.}), in the context of a network with at most $N$ peers, where the average block interval is $\beta$.
\begin{enumerate}
    \item We design protocols to maintain a dynamic hypercubic network of $\Theta(\log N)$-sized committees, where the half-life is $\alpha = \Theta(\beta \sqrt{N} \log N)$ and the network size can vary polynomially over time.
    \item We prove that the graph formed by the honest peers remains connected for polynomial number of rounds, and each peer sends/receives only $O(\log^3 N)$ messages per round.
    \item We show that even when catastrophic failures occur (e.g., a constant fraction of committees and their corresponding directory members are instantly corrupted), the overlay recovers within a small number of half-lives.
    \item We give a lower bound, barring log-factors, for minimum feasible half-life, $\alpha = \widetilde{\Omega}(\sqrt{\beta N})$, showing that it is impossible to tolerate higher rates of churn, even if peers share a public bulletin board that can be used for joining.
\end{enumerate}

\section{Model}\label{sec:model}

\noindent\textbf{Entities.} A peer is the real-world entity that participates in the blockchain network. There are two types of peers: (1) honest peers that follow the specified protocols, and (2) Byzantine peers that may deviate from the protocols in an arbitrary way. The network size refers to the total number of peers within the network. The \emph{maximum} network size is denoted by $N$, though a peer can control multiple \textit{identities} or \textit{nodes} within the network.

\noindent\textbf{Communication.} The system proceeds in synchronous \textit{rounds}; in each round, in addition to local computation, a message that is sent at the beginning of a round by a peer is assumed to reach its neighbours by the end of that round. Each peer maintains a set of \textit{neighbouring} peers that it is said to be \emph{connected} with, i.e., those peers are used for sending and receiving information.

\noindent\textbf{Computational restriction.} The peers are associated with a \textit{hash power} constraint, i.e., each peer owns one unit of hash power that allows the peer to query a hash function (modelled as a random oracle) $q > 1$ times in a round \cite{garay2015bitcoin, pass2017analysis}. (If some entity has more hash power, then it can be viewed as a coalition of peers.) Given an input of any length, the hash function is assumed to provide a random output of (fixed) length $\kappa = \Theta(\log N)$.

\noindent\textbf{Blockchain.} If there is a graph (where vertices are honest peers and edges are connections) of at least $\mu_n (1-\rho) n$ honest peers\footnote{We consider a (large) subset of honest peers because the network can get split into multiple components during a catastrophic failure (cf. Section \ref{sec:recovery}).} with a diameter of $\Delta \leq 2 \log N$, and there are at most $\rho n$ Byzantine peers, for appropriate constants $\rho < \mu_n < 1$, where $n$ is the \emph{current} number of peers, then the blockchain is guaranteed to provide the following properties for those honest peers with probability at least $1-2^{-\kappa}$.

    \textit{Safety.} There exists a notion of confirmed chain\footnote{We do not delve into details of how to obtain this confirmed chain as this may be blockchain-specific, for example, Bitcoin deems a block to be confirmed if it is at least 6 blocks deep.} $C^r_u$ for any honest peer $u$ in round $r$.
     \begin{itemize}
        \item ($\Delta$-Synchronization) If $|C^r_u|$ is the length of the confirmed chain of honest peer $u$ at round $r$, then by round $r + \Delta$, every honest peer's confirmed chain's length is at least $|C^r_u|$.
        \item (Consistency) $C^r_u$ is always a prefix of $C^{r'}_u$ for any round $r' \geq r$. At any round $r$, if $|C^r_u| \leq |C^{r}_v|$, then $C^r_u$ is a prefix of $C^{r}_v$, and for a large enough constant $\mu_s$, $|C^r_u| - |C^{r}_v| \leq \mu_s$.
    \end{itemize}
    
    
    \textit{Liveness.} For large enough $T$ and constant $\mu_b$, any consecutive $T \geq \Omega(1)$ blocks are included in any confirmed chain in $[T\beta / \mu_b, T\mu_b\beta]$ rounds, where $\beta$ is said to be the average block interval.
    
    
    \textit{$\delta$-Approximate Fairness.} Any set of honest peers controlling a $\phi$ fraction of hash power own at least a $(1 - \delta)\phi$ fraction of the blocks in any $\Omega(\kappa/\delta)$ length segment of the chain \cite{pass2017fruitchains}.
    
    \textit{Public Availability.} There exists an introductory service $\mathcal{I}$ that provides a copy of an existing honest peer's blockchain. (In practice, this typically means that there exist a set of publicly available blockchains, of which at least some are honest. The analysis holds if the public copies are slightly outdated, say by a constant number of blocks.)

\noindent\textbf{Adversary.} We consider Byzantine peers that can collude and arbitrarily deviate from the specified protocols. In any round, the number of Byzantine peers is at most a $\rho$ fraction of the network size. They know the entire network topology in any round, but they cannot modify or delete messages sent by honest peers. The goal of Byzantine peers is to disturb the normal functioning of the overlay, for e.g., isolate a subset of honest peers by occupying all their connections, or increase the diameter of the overlay, etc.


\noindent\textbf{Churn.} We consider an adversary \cite{awerbuch2009towards, guerraoui2013highly} that specifies the join/leave sequence $\sigma$ for honest peers in advance. But it can choose to adaptively join/leave a Byzantine peer. In particular, after the first $i$ events in $\sigma$ are executed, the adversary can either choose to join/leave a Byzantine peer or initiate the $(i+1)^{\mathrm{th}}$ event in $\sigma$. (This models scenarios such as an honest peer $h$ stays for a long time, and the other honest peers are subjected to churn, and Byzantine peers can adaptively rejoin until $h$'s neighbouring peers are Byzantine.)

\noindent\textbf{Churn rate.}\label{SUBSECchurnass} We consider the \emph{half-life} measure \cite{liben2002analysis} to model the churn rate for honest peers. Formally, at any given round $r$, the halving time is the number of rounds taken for half the number of honest peers (which were alive at round $r$) to leave the network; similarly, the doubling time is the number of rounds required for the number of honest peers to double.

An \textit{epoch}, denoted by $\alpha$, is defined as the smallest halving time or doubling time over all rounds in the execution. Furthermore, we assume that the epoch is much greater than the average block interval; in other words, $\alpha \gg \beta \log N$.

\noindent\textbf{Honest peer failure.} By the churn rate assumption, at most half of the honest peers that are alive at round $r$ can leave (fail) by the end of $r + \alpha$ rounds. (We do not make any distinction between a leave event and a failure.) Since the adversary has to obliviously specify the join/leave sequence for honest peers (and they are initially assigned random identities), each honest peer is assumed to (independently) fail with a probability $p_{f} \leq 0.5$ in each epoch, where $p_{\mathit{f}}$ is the fraction of honest nodes that leave during the epoch.

\noindent\textbf{Change of network size.} The network size can change by a factor of at most 2 in any epoch. It can polynomially vary over time; the number of peers at any round $r$, $N_r \in [N^{1/y}, N]$ for some constant $y > 1$.

\section{Stable Network Size}\label{sec:stablenetsize}
In this section, we assume that the total number of peers is fixed and equal to $\Theta(N)$. In Section \ref{sec:dynamic}, this assumption is relaxed, where we allow a polynomial variation in network size over time. (For simplicity, we avoid using floor/ceiling repeatedly unless necessary.) We describe the protocols for a peer by making use of the access to its confirmed chain. We now describe the overlay structure and our high-level approach.

\noindent\textbf{Nodes.} Each peer generates and controls $\Theta(\log N)$ (virtual) \textit{nodes} in the overlay. Each peer participates (sends and receives messages) in the network through its nodes. All peers are required to perform proof-of-work to generate nodes. Each node has a \textit{lifetime} after which it will be considered invalid (or removed). There are two \textit{roles} of a node: a (1) directory node, and a (2) non-directory node. Every node is initially a non-directory node, but some nodes become directory nodes as well, playing both roles.

\noindent\textbf{Committees.} Our protocols maintain a structured network of \textit{committees}, specifically, a hypercube whose vertices correspond to committees. In total, there are $\mathcal{C} = N$ committees, each consisting of $\Theta(\log N)$ nodes. (Logarithmic redundancy is used to show that each committee has a sufficient number of honest peers amidst churn, so that the hypercube structure is maintained.) Each committee is identified by a (unique) committee ID of $\log N$ bits. As a peer may be controlling multiple nodes, it may be present in multiple committees.

\noindent\textbf{Directory node.} A peer that successfully adds a block to the blockchain promotes one of its nodes to a directory node. While trying to mine for blocks, each peer adds its network address into the prospective block. The directory nodes are responsible for helping incoming new nodes to join a (random) committee. (Recall that a new node can contact directory nodes as a copy of the blockchain is publicly available.) They do so by providing the network addresses of the required committee members.

\pg{High-level overview.} Byzantine peers can repeatedly join and leave until they get placed in a specific set of committees. This can result in most of an honest node's neighbours being Byzantine over time, as other honest neighbours can leave (due to churn). We enforce a limited lifetime (that is close to the half-life of honest peers) for all nodes so that an honest node would have to leave and rejoin some other (random) committee, whilst some of its neighbours are still honest. The technical difficulty lies in designing a secure bootstrapping mechanism while handling churn, bandwidth-constraints and Byzantine interference.

\subsection{Directories}\label{SUBSECdir}
Directory nodes are critical for maintaining membership information of nodes in the network. A directory comprises $\mathcal{B}$ consecutive ``buckets'', and each bucket consists of consecutive $\lambda_d\log^2 N$ blocks, where $\lambda_d$ is a suitable constant. Each peer that creates a block added to the confirmed chain becomes part of one of the buckets in the directory (via one of its nodes). A directory node can be associated with its block number. The total number of directory nodes is equal to $\mathcal{K} = \mathcal{B}\lambda_d\log^2 N$.

Since directory nodes help new peers join the network, we need to show that there always exists enough honest directory nodes in each bucket amidst churn. Using the fairness assumption, it turns out that $\Theta(\log^2 N)$ bucket size is sufficient to show that there are $\Omega(\log N)$ honest peers.

\noindent\textbf{Functions.} Each bucket is \textit{responsible} for a set of committees, i.e., all the directory nodes in that bucket are supposed to help new nodes join a specific set of committees. There are two main functions of a directory node.
\begin{enumerate}
    \item A directory node \textit{stores} information about nodes belonging to a set of committees. Specifically, if a new node joins one of those committees, then it stores the new node's \textit{entry information}\footnote{See \texttt{JOIN} protocol description in Section \ref{subsec:join} for the exact definition. The word ``entry'' may be dropped when the context is clear. A committee's information refers to the set of all its nodes' entry information.} (that includes network address, etc.).
    \item A directory node \textit{sends} information about committees that the directory node is responsible for. It sends entry information about all the existing nodes in the committee that the new node belongs to, or its neighbouring committees.
\end{enumerate}

\noindent\textbf{Phases of a bucket.} The buckets can belong to one of the following four phases; the directory nodes are also classified into those four phases in the same way as buckets. We describe a directory node's functions at each phase, since its inception in the confirmed chain. This is summarized in Table \ref{tab:phases}. (The chain is divided into buckets since the beginning, resulting in periods where the chain length may not be an integral multiple of $\lambda_d\log^2 N$. Except in infant phase, a bucket in every other phase has all the $\lambda_d\log^2 N$ blocks confirmed.)
\begin{enumerate}
    \item \textit{Infant.} A bucket in which at least one block (out of $\lambda_d\log^2 N$) is confirmed, but not all the $\lambda_d\log^2 N$ blocks are confirmed yet. The nodes do not store the incoming new nodes' entry information. They do not respond to the incoming new nodes about any committees.
    \item \textit{Middle-aged.} These nodes store incoming new node entry information, and reply to them about the relevant committees' information that they know about.
    \item \textit{Veteran.} These nodes do not store new nodes' entry information, but they do reply about the relevant committees' information known to them.
    \item \textit{Dead.} These nodes (as in infant phase) neither store any new node information, nor reply to new nodes about any committee information.
\end{enumerate}

\pg{Phase transitions.} As new blocks are added, new directory nodes take over the functions of old directory nodes. For a peer within the network, we detail the transitioning of a directory node from one phase to another. (In the transitions, there is a delay of $\Delta$ rounds is to ensure that the confirmed chains of honest peers reach the same height before making those transitions.) 
\begin{enumerate}
    \item \textit{Infant to Middle-aged.} If all $\lambda_d\log^2 N$ blocks of the bucket are confirmed in the blockchain, then the bucket transitions to middle-aged phase.
    \item \textit{Middle-aged to Veteran.} If the bucket is not among the most recent (confirmed) $\mathcal{B}$ buckets, then the bucket transitions to veteran phase after a delay of $\Delta$ rounds.
    \item \textit{Veteran to Dead.} If the bucket is not among the most recent (confirmed) $\mathcal{B}_{\mathit{act}}$ buckets, then the bucket transitions to dead phase after waiting for $\Delta$ rounds.
\end{enumerate}

\noindent\textbf{Committee-Directory mapping.} There exists a predetermined mapping $\mathcal{M}: [\mathcal{C}] \rightarrow [\mathcal{B}]$ from committees to buckets for a directory. The mapping is set such that each bucket is responsible for (almost) the same number of committees, and the sets of committees that any two buckets (in the same directory) are responsible for, are disjoint. (This is done so that the join requests are load balanced across the directory.) This mapping is useful for new nodes joining the network, to know which directory nodes to contact, to get information about relevant committees.

\noindent\textbf{Active directory.} The \emph{active} directory, also known as \emph{bootstrapping service}, is the most recent $\mathcal{B}_{\mathit{act}}$ consecutive (confirmed) buckets. Due to phase transitioning, the most recent $\mathcal{B}$ consecutive (confirmed) buckets, forming an entire directory, are middle-aged. The rest of the buckets, forming one or more directories, are veteran. This means that there are multiple buckets (across different directories) that are responsible for a given set of committees (but only one of them being in middle-age phase). The number of blocks and number of buckets in an active directory are denoted as $\mathcal{K}_{\mathit{act}}$ and $\mathcal{B}_{\mathit{act}}$ respectively.

New nodes figure out the sequence of buckets in the active directory (using the introductory service's chain, block numbers and committee-directory mapping), and contact the relevant buckets to get the required committees' entry information for joining the network. The delay of $\Delta$ rounds in transitions makes up for any lag in the chain provided by the introductory service. An example of an active directory with its buckets in different phases is illustrated in Figure \ref{fig:algo-wrk}.


\begin{table}
\begin{center}
\begin{tabular}{ |c|c|c| } 
\hline
 Phase & Store entry info & Reply entry info \\ 
\hline
 Infant & \xmark & \xmark\\
\hline 
 Middle-aged & \cmark & \cmark\\
\hline 
 Veteran & \xmark & \cmark\\
\hline
 Dead & \xmark & \xmark\\
 \hline
\end{tabular}
\caption{Phases and functions of a directory node.}
\label{tab:phases}
\end{center}
\end{table}

\begin{figure}[tb]
\centering
\includegraphics[width=1\linewidth]{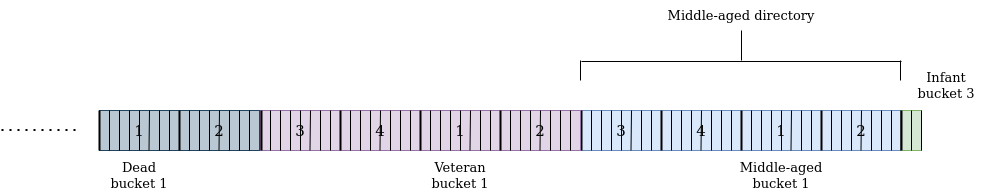}
\caption{An illustration of an active directory of length equal to two directories, wherein each directory has four buckets. The bucket numbers are denoted inside the bucket. The blue buckets, which are the most recent 4 (fully formed) buckets in the blockchain, are middle-aged. They both store and respond with committee information. The purple buckets, which are the next 4 buckets in the active directory, are veteran. They only respond with committee information. The grey buckets and the ones after them, are dead and do not participate in bootstrapping of new nodes. The green blocks being formed towards the end, are part of the infant bucket. They start functioning once they become middle-aged.}
\label{fig:algo-wrk}
\end{figure}

A new node must provide a \textit{proof} for joining the network. The directory nodes only interact with a new node if its proof is valid. Section \ref{subsec:join} provides details on the proof for joining and storing/sending committee entry information to a new node. Algorithm \ref{alg:dir} provides a succinct description of the protocol followed by a directory node, from Infant to Dead phase, including the functions and transitions. The pseudocodes for the subroutines \texttt{VERIFY\_PROOFS}, \texttt{STORE\_INFO} and \texttt{REPLY\_INFO} are given in Appendix \ref{appendix:subroutines}, though they are self-explanatory in the given context.

\begin{algorithm}
\caption{\texttt{DIR} protocol}
\label{alg:dir}
\begin{algorithmic}[1]
\REQUIRE A peer runs this protocol for its node that is embedded in block $b$ in its confirmed chain, after $b$ becomes part of a bucket $\mathit{bkt}$ (end of infant phase). Let $M_1$ be the set of all \texttt{JOINING} messages received in a round. Let $M_2$ be the set of all \texttt{REQ\_INFO} messages received in a round.
\ENSURE Contribute to the bootstrapping service via its directory node.
\WHILE{$\mathit{bkt} \in $ most recent confirmed $\mathcal{B}$ buckets}
\STATE $V_1 \leftarrow$ \texttt{VERIFY\_PROOFS}($M_1$), $V_2 \leftarrow$ \texttt{VERIFY\_PROOFS}($M_2$).
\STATE \texttt{STORE\_INFO}($V_1$).
\STATE \texttt{REPLY\_INFO}($V_2$).
\ENDWHILE
\FOR{$\Delta$ rounds}
\STATE $V_1 \leftarrow$ \texttt{VERIFY\_PROOFS}($M_1$), $V_2 \leftarrow$ \texttt{VERIFY\_PROOFS}($M_2$).
\STATE \texttt{STORE\_INFO}($V_1$).
\STATE \texttt{REPLY\_INFO}($V_2$).
\ENDFOR
\WHILE{$\mathit{bkt} \in $ most recent confirmed $\mathcal{B}_{\mathit{act}}$ buckets}
\STATE $V \leftarrow$ \texttt{VERIFY\_PROOFS}($M_2$).
\STATE \texttt{REPLY\_INFO}($V$).
\ENDWHILE
\FOR{$\Delta$ rounds}
\STATE $V \leftarrow$ \texttt{VERIFY\_PROOFS}($M_2$).
\STATE \texttt{REPLY\_INFO}($V$).
\ENDFOR
\end{algorithmic}
\end{algorithm}

\subsection{Node joins and lifetimes}\label{subsec:join}

Our network utilizes proof-of-work as Sybil defense to limit the number of nodes controlled by any peer in any round, requiring peers to continuously mine for nodes, in addition to blocks.


\noindent\textbf{Proof for joining.} Let $N_c$ be the nonce, $\hat{\mathrm{B}}_l$ be the hash of the latest confirmed block $\mathrm{B}_l$, and $\mathit{net\_addr}$ be the network address of the peer. Then, the peer evaluates, $P_{\mathit{join}} = \mathbf{H}(\hat{\mathrm{B}}_l \mathbin\Vert \mathit{net\_addr} \mathbin\Vert N_c)$,
to join the network through a new node. If $P_{\mathit{join}} < T_{\mathit{join}}$, where $T_{\mathit{join}}$ is the \textit{mining target} for joining the network, then the node is considered to be a ``valid'' node, which means that the peer would be able to communicate with the bootstrapping service to register that node, and join the network. (A directory node rejects the proof if $\mathrm{B}_l$ is not among the most recent $\mu_s$ blocks in its confirmed chain, as every pair of honest peers' confirmed chains differ by at most $\mu_s$ blocks.)

\noindent\textbf{Joining the network.} A node's \textit{entry information} constitutes its network address, the nonce $N_c$ and the block number of the block that was used while mining for that node. Recall that a new peer gets a copy of the blockchain by the introductory service. We now describe the steps taken by a peer $p$ to generate and join a new node $q$ into the network. See Algorithm \ref{alg:join} for the pseudocode of \texttt{JOIN} protocol.
\begin{enumerate}
    \item \textit{Solve puzzle.} The peer expends computational resources to generate a valid node $q$. The leftmost $\log N$ bits of $P_{\mathit{join}}$ represent the ID of the (random) committee, denoted by $c$, to which this new node would belong to. Let $C_{\mathit{rel}}$ be the set of neighbouring committees of $c$, including $c$.
    \item \textit{Store info.} Let $b^c_m$ be the middle-aged bucket responsible for committee $c$. Peer $p$ sends (\texttt{JOINING}, entry information of $q$) to all nodes in $b^c_m$. (The directory nodes in $b^c_m$ store $q$'s entry information.)
    \item \textit{Request info.} Let $B^c$ be set of all buckets responsible for committee $c$. For each $k$ in $C_{\mathit{rel}}$ and each $b$ in $B^c$, peer $p$ sends (\texttt{REQ\_INFO}, $k$, entry information of $q$) to $\lambda_{j} \log N$ nodes that are sampled uniformly and independently from $b$, where $\lambda_j$ is a suitable constant. (For each $k$ in $C_{\mathit{rel}}$ and each $b$ in $B^c$, a directory node in $b$ sends a \texttt{COMM\_INFO} message consisting of its knowledge of entry information of nodes in $k$, if it received a \texttt{REQ\_INFO} message from $q$.)
    \item \textit{Handle $\Delta$-synchrony.} Let $b_1$ and $b_2$ be the first and $(\mathcal{B} + 1)^{\mathrm{th}}$ confirmed buckets (i.e., most recent middle-aged and veteran buckets). If the number of blocks confirmed after bucket $b_1$ is at most $\mu_s$, and if $b_1$ is responsible for committee $c$, then send (\texttt{JOINING}, entry information of $q$) to all nodes in $b_2$.
    \item \textit{Join committee.} Node $q$ takes the union of entry information in the \texttt{COMM\_INFO} messages for each committee (as the adversary can only underrepresent the nodes in a committee). Node $q$ sends (\texttt{JOINING}, entry information of $q$) to the nodes in each committee $c$ in $C_{\mathit{rel}}$.
\end{enumerate}

The key observation is that a constant fraction of directory nodes in a bucket are honest and available at any time, due to the logarithmic redundancy in buckets and blockchain fairness. Thus, it is sufficient for the new node needs to hear back information from $O(\log N)$ directory nodes in each bucket (Step 3), reducing the communication complexity of join down to $O(\log^3 N)$, i.e., the new node contacts at most $\log N$ buckets, wherein $O(\log N)$ directory nodes in each bucket reply with committee information of $O(\log N)$ nodes.

\begin{algorithm}
\caption{\texttt{JOIN} protocol}
\label{alg:join}
\begin{algorithmic}[1]
\REQUIRE Chain provided by introductory service $\mathcal{I}$. Let $B^c$ be set of all buckets responsible for a committee $c$. Let $b^c_m$ be the middle-aged bucket responsible for a committee $c$. Let $b_1$ and $b_2$ be the first and $(\mathcal{B} + 1)^{\mathrm{th}}$ confirmed buckets. Let $\mathbf{H}(.)$ denote the hash function.
\ENSURE To generate and register a node in the active directory, and join the network.
\STATE $P_{\mathit{join}} \leftarrow \infty$.
\STATE $N_c \leftarrow 0$.
\WHILE{$P_{\mathit{join}} \geq T_{\mathit{join}}$}
\STATE $\mathrm{B}_{l} \leftarrow$ Most recent confirmed block of the chain given by $\mathcal{I}$.
\STATE $\hat{\mathrm{B}}_l \leftarrow \mathbf{H}(\mathrm{B}_{l})$.
\STATE $P_{\mathit{join}} \leftarrow
\mathbf{H}(\hat{\mathrm{B}}_l \mathbin\Vert \mathit{net\_addr} \mathbin\Vert N_c)$.
\STATE $N_c \leftarrow N_c + 1$.
\ENDWHILE
\STATE $c \leftarrow $ Leftmost $\log N$ bits of $P_{\mathit{join}}$.
\STATE $C_{\mathit{rel}} \leftarrow $ Set of neighbouring committees of committee $c$, including $c$.
\STATE SEND (\texttt{JOINING}, entry information of $q$) to all nodes in $b^c_m$.
\IF{number of blocks confirmed after bucket $b_1 \leq \mu_s$}
\IF{$b_1$ is responsible for a committee in $C_{\mathit{rel}}$}
    \STATE SEND (\texttt{JOINING}, entry information of $q$) to all nodes in $b_2$.
\ENDIF
\ENDIF
\FOR{each committee $k$ in $C_{\mathit{rel}}$}
\FOR{each bucket $b$ in $B^c$}
    \STATE $R \leftarrow $  $\lambda_{j} \log N$ nodes sampled uniformly and independently from $b$.
    \STATE SEND (\texttt{REQ\_INFO}, $k$, entry information of $q$) to all nodes in $R$.
\ENDFOR
\ENDFOR
\FOR{each committee $c$ in $C_{\mathit{rel}}$}
    \STATE RECEIVE (\texttt{COMM\_INFO}, $c$) from each $B^c$.
    \STATE $U \leftarrow $ Union of entry information in the received \texttt{COMM\_INFO} messages.
    \STATE SEND (\texttt{JOINING}, entry information of $q$) to all nodes in $U$.
\ENDFOR
\end{algorithmic}
\end{algorithm}

\pg{Perturbing the network.} To avoid Byzantine peers repeatedly rejoining to populate a specific set of committees, while exploiting the churn of honest peers, a standard solution is to employ a limited lifetime for the nodes (and force them to rejoin) \cite{awerbuch2004group}. As this would keep the system in a hyperactive state, the node lifetime should be as large as possible. However, if the node lifetime is too large compared to the half-life, then the adversary may be able to isolate peers. Thus, we set the node lifetime to be $\Theta(\alpha / \beta)$ blocks, which amounts to $\Theta(\alpha)$ rounds (constant number of half-lives) due to blockchain liveness. This helps us show that at most a constant fraction of any honest peer's neighbours are Byzantine at any time.

\pg{Entry time.} Peers attach the hash of the most recent confirmed block $\mathrm{B}_{l}$ (given by the introductory service) as part of the input to the hash function when they are mining a new node. The proof for joining conveniently records the new node's ``entry time'' (in terms of block number). This helps other peers to keep track of a node's ``age'' as they would have received its entry information in their first interaction with the node.

\noindent\textbf{Lifetime of non-directory node.} The lifetime of a non-directory node is $T_{\mathit{l}} =  \lambda_{l}\alpha/\beta$ blocks, where $\lambda_{l}$ is a suitable constant. The node $u$ that had joined at block $b_{l}$ would be considered invalid after block $b_{l} + T_{\mathit{l}}$ is confirmed, where $b_l$ is the block number of block $\mathrm{B}_{l}$, at which point the peer stops controlling that node $u$, and all the other peers that had node $u$ as its neighbour remove $u$ from their nodes' neighbour list.

\noindent\textbf{Lifetime of directory node.} If a node is promoted to a directory node, then it obtains another life (separate from the non-directory life). When a node becomes a directory node, it continues to perform the functions of a non-directory node as long as the non-directory role is considered to be valid. The directory node is considered to be alive for $T_{\mathit{dl}}$ blocks since the block in which it is embedded in. We set $T_{\mathit{dl}}$ to be $\lambda_{\mathit{dl}}\alpha / \beta$ blocks, where $\lambda_{l} < \lambda_{\mathit{dl}}$ is some constant.

A directory node's lifetime is determined by how long it needs to stay in each of its phases. Firstly, if a directory node is the first node to be part of a bucket, then it needs to be alive for $\lambda_d\log^2 N$ blocks (size of a bucket). Secondly, a directory node stays in the middle-aged phase until the next bucket that is responsible for the same set of committees is formed, and this takes $\mathcal{K}$ blocks. Finally, a directory node needs to be in veteran phase only until some (non-directory) node in one of the committees that it is responsible for, is still valid. And this takes at most $T_{\mathit{l}}$ blocks. The lifetime of directory node $T_{\mathit{dl}}$ is slightly more than $\mathcal{K}_{\mathit{act}}$ because there is a delay of $\Delta$ rounds in the transitions. Thus, $T_{\mathit{dl}} > (1 + \mathcal{B})\lambda_d \log^2 N + T_l$.

\pg{Node generation rate.} Let $p_{n}$ be the \textit{difficulty threshold} for node mining. It is the probability that a \emph{single} hash query is successful in generating a valid node. $T_{\mathit{join}}$ is set such that, in each epoch, the expected number of valid nodes that can be generated is equal to $p_{n}q\alpha N = \lambda_{n}N\log N$, where $\lambda_n$ is a suitable constant, so that each committee has $\Theta(\log N)$ nodes at any time. This is because in every epoch, about $\Theta(N \log N)$ nodes join, while a similar number of them leave due to limited lifetime (set to be a constant number of epochs). 

\subsection{Setting parameters} \label{subsec:setparam}

\noindent\textbf{Constraints on $\alpha$ and $\mathcal{B}$.} The first constraint between $\alpha$ and $\mathcal{B}$ arises due to the lifetime of a non-directory node,
\begin{equation}\label{eq:buck-upper-bound}
    \mathcal{B} \leq \Theta\left( \frac{\alpha}{\beta\log^2 N} \right).
\end{equation}

Furthermore, a natural constraint on $\alpha$ and $\mathcal{B}$ arises due to the entry of new nodes. Due to node generate rate, the system must be able to allow $\lambda_{n}N \log N$ nodes to enter the network in any $\alpha$ (consecutive) rounds. As a new node sends the (same) join request to all directories within the active directory (as it contacts both middle-aged and veteran buckets), it suffices to focus on the number of join requests handled by one directory. (Here, ``join request'' can be \texttt{JOINING} or \texttt{REQ\_INFO} message of a valid new node.)

Due to the Sybil defense mechanism, we can ensure that the join requests are load-balanced across all the rounds in an epoch. Let $\lambda_{\mathit{jr}}$ be the highest number of join requests that can be handled by a bucket per round. For any directory, we calculate the total number of join requests that need to be handled in any round to be
\begin{equation}\label{eq:buck-lower-bound}
    \mathcal{B}\lambda_{\mathit{jr}} \geq \Theta\left(\frac{\beta N \log^2 N}{\alpha}\right),
\end{equation}
where LHS of the inequality is the total number of join requests that can be handled in a round, and RHS represents the minimum number of join requests that need to be handled in a round. We now provide an explanation for the extra $\beta$ and $\log N$ factors, i.e., the half-life and number of buckets in a directory are appropriately increased to handle some extra join messages such that for any peer, the total communication complexity per round is $O(\log^3 N)$. (Refer to Lemma \ref{lemma:DIRbandwidthcost} for the full details.)
\begin{itemize}
    \item Recall that a directory node accepts the proof if the block used for the puzzle is among the $\mu_s$ most recent confirmed blocks. Thus, the adversary can launch a pre-computation attack of join messages generated over a period of constant number of block intervals, which amounts to about $\Theta(\beta)$ rounds (due to blockchain liveness).
    \item The extra $\log N$ factor arises due to new nodes having to contact $O(\log N)$ buckets while joining.
\end{itemize}

\noindent\textbf{Setting $\alpha$ and $\mathcal{B}$.} Ideally, the overlay needs to be robust for the smallest possible half-life, so that the network can withstand the limits of churn. Due to the aforementioned constraints, and to ensure a bandwidth cost of $O(\log^3 N)$ messages per round for a peer for overlay maintenance, we get that $\alpha = \Theta(\beta \sqrt{N} \log N)$ and $\mathcal{B} = \Theta(\sqrt{N} / \log N)$. It may seem that the churn rate is rather low, but in Section \ref{sec:lower-bound}, we show that this value of $\alpha$ is close to optimal. Specifically, we show that any dynamic system where peers have a bandwidth-constraint of $\polylog(N)$ messages per round, that uses a blockchain as an entry point in bootstrapping, must have $\alpha = \widetilde{\Omega}({\sqrt{\beta N}})$.

\subsection*{Bootstrapping network} \label{subsec:bootnetwork}
In this work, we only focus on maintaining the overlay, and make an assumption that at time zero, the network has formed a hypercube structure of committees, with $\Theta(\log N)$ honest peers in each committee, and $\Theta(\log^2 N)$ honest nodes in each bucket of the first (active) directory, and each peer controlling at most $O(\log N)$ nodes. If $N_0$ is the number of peers in the beginning of the execution, then the number of committees $\mathcal{C} = N_0$. In each committee, the honest nodes are connected to each other, and each honest node is connected to $\Omega(\log N)$ honest nodes in each neighbouring committee. Each directory node knows entry information of the required committees. Moreover, the node difficulty threshold is appropriately set, $p_n = (\lambda_n \log N) / (q \alpha)$. At this point, the network allows dynamic participation and satisfies the required properties.

\section{Dynamic Network Size}\label{sec:dynamic}

In this section, we augment the protocols described in Section~\ref{sec:stablenetsize} to handle changing numbers of peers. The main problem caused by polynomial variation in network size is that the peer redundancy factor in committees gets affected, i.e., the system needs to adapt to maintain logarithmic redundancy. One way to deal with the problem, is to keep the number of committees fixed and change the node difficulty threshold, $p_n$. But if the network size keeps decreasing, then each peer would need to simulate too many nodes at any time in order for the system to maintain $\Theta(\log N)$ nodes in every committee. And if the network size keeps increasing, then some peers may not be able to participate in the network all the time because they may take too long to generate nodes. Therefore, we fix $p_n = (\lambda_n \log N) / (q \alpha)$, and provide protocols to switch to a new hypercube of a different dimension.

Our key insight is that the switch to a new topology can be carried out in a straightforward manner by exploiting the global coordination provided by the blockchain. We begin by providing a high-level intuition of our approach. The first step is to efficiently estimate the network size in every constant number of epochs (cf. Section \ref{subsec:netsizeest}). First, each node keeps track of new nodes that joined its committee in a span of fixed number of blocks. Then, all the nodes of a (random) committee broadcast the entry information of all new nodes that joined that committee. Finally, using a balls-and-bins argument, all the peers can get a good estimate of the network size, and they simply use the blockchain to agree on it. This estimate is used for the next two components of handling varying network size: (1) resetting overlay parameters and (2) changing dimension of the hypercube.

Each peer computes the overlay parameters, namely, trigger for dimension change and number of committees (in the next hypercube), using the network size estimate. Using the blockchain fairness property, all the peers arrive at a consensus on the parameters by taking the majority over $\Theta(\log N)$ consecutive blocks (cf. Section \ref{subsec:resetparam}). During dimension change, the directory goes into a ``split state'' for about an epoch, serving committees in the current hypercube and also constructing committees in the next hypercube (i.e., new nodes also get placed in the next hypercube). Until each committee in the next hypercube has sufficient number of new (honest) nodes, the overlay operates in the old hypercube. The directory then stops serving committees in the old hypercube, and the network adopts the new hypercube for broadcasting blocks (cf. Section \ref{subsec:dimchange}).

\subsection*{Time in terms of b-epochs}
Our protocols, in contrast to the ones described in the previous section, require considerable coordination among the peers at regular intervals. To describe the protocols, we divide a period of $\Theta(\alpha)$ rounds into two ``phases''. Phase 1 is called the \textit{estimation} phase, where the existing peers calculate an estimate of the total number of nodes that joined in that phase to get a good estimate of the network size. There are $\alpha_1$ rounds in phase 1. Phase 2 is called the \textit{agreement} phase, where the peers reach an agreement (via blockchain) on the new difficulty threshold and the decision to change dimension. There are $\alpha_2$ rounds in phase 2. These two phases make up the epoch, $\alpha = \alpha_1 + \alpha_2$. The length of phase 2 is actually much smaller than the length of phase 1 (by design, see Section \ref{subsec:resetparam}), i.e., $\alpha_2 \ll \alpha_1$; in particular, $\alpha_1 = \Theta(\alpha)$.

Analogous to an epoch, we say that $\alpha / (\mu_b \beta)$ consecutive blocks is called a \textit{b-epoch}. The number of rounds elapsed during that period of block intervals is $\left[\frac{\alpha}{\mu^2_b \beta}, \alpha\right]$ due to blockchain liveness. Similarly, the phase 1 of a b-epoch consists of $\alpha_1 / (\mu_b \beta)$ consecutive blocks. The peers rely on the most recent confirmed block number to establish: (1) the start of a b-epoch, and (2) the end of the first phase of a b-epoch. The idea is to fix multiples of an appropriate block number to mark the start and the end of phase 1 of a b-epoch. Thus, using blockchain safety, the honest peers can run the protocols at (approximately) the same time, albeit at most $\Delta$ rounds delay.

\subsection{Network size estimation} \label{subsec:netsizeest}
The peers require a good estimate of the network size for changing the dimension of the hypercube. This is done in phase 1 of every b-epoch. Algorithm \ref{alg:sizeest} provides the pseudocode of \texttt{NET\_SIZE\_EST} protocol.
\begin{enumerate}
    \item Each node keeps track of the set of (new) nodes that join its committee from start of a b-epoch $e$ to end of phase 1 of b-epoch $e$.
    \item Let $b^k_e$ be the block that marks the end of phase 1 of b-epoch $e$. Let the leftmost $\lceil\log \mathcal{C}_e \rceil$ bits of $\mathbf{H}(b^k_e)$ determine a random committee $s$ (due to random oracle assumption), where $\mathcal{C}_e$ is the number of committees in b-epoch $e$.
    \item All the nodes belonging to committee $s$ broadcast the entry information of nodes that joined committee $s$ in phase 1.
    \item The peers in the other committees take the union of the responses to obtain the set of nodes that joined committee $s$ in phase 1. Let $H_e$ be the total number of those nodes.
    \item $G'_e = H_e\mathcal{C}_e$ is the estimate of total number of new nodes that joined in phase 1. Using this estimate, the nodes calculate the network size estimate $M'_e = \mu_b G'_e / (p_n \alpha_1)$.
\end{enumerate}

\begin{algorithm}
\caption{\texttt{NET\_SIZE\_EST} protocol}
\label{alg:sizeest}
\begin{algorithmic}
\REQUIRE Start this protocol at the beginning of a b-epoch $e$. Let $l$ and $k$ denote the block numbers of the most recent confirmed block and block that marks the end of phase 1 of b-epoch $e$ respectively. Let $b^t$ denote the block in the confirmed chain with block number $t$. Let the node $u$ running this protocol belong to committee $c$. Let $M$ be the set of all \texttt{JOINING} messages received in a round. 
\ENSURE Output $M'_e$ such that $M'_e \in [(1-\delta_{\mathit{err}})(1-\rho)M^L_e / \mu_b, (1+\delta_{\mathit{err}})\mu_b M^H_e]$ where $M^L_e$ and $M^H_e$ are minimum and maximum network sizes in phase 1 of b-epoch $e$, and $\delta_{\mathit{err}} < 1$ is a small positive constant.
\STATE $C_u \leftarrow \{ \}$.
\WHILE{$l < k$}
\STATE $V \leftarrow$ \texttt{VERIFY\_PROOFS}($M$).
\STATE $C_u \leftarrow C_u \bigcup V$. \COMMENT{Store entry information of nodes that join committee $c$.}
\ENDWHILE
\STATE $s \leftarrow$ Leftmost $\lceil \log \mathcal{C}_e \rceil$ bits of $\mathbf{H}(b^{k})$. \COMMENT{Pick random committee.}
\STATE Wait for $\Delta$ rounds. \COMMENT{All honest peers reach end of phase 1.}
\IF{$s$ is $c$}
\STATE $H_e \leftarrow |C_u|$. \COMMENT{Number of phase 1 node joins in committee $c$.}
\STATE BROADCAST (\texttt{EST\_INFO}, $C_u$).
\ELSE
\STATE Wait for $\Delta$ rounds. \COMMENT{Wait for the broadcast to be completed.}
\STATE RECEIVE (\texttt{EST\_INFO}, $C_v$) where $C_v$ is a set of \texttt{JOINING} messages broadcasted by some node $v$ in committee $s$.
\STATE $C^s \leftarrow \bigcup C_v$. \COMMENT{All phase 1 node joins in committee $s \neq c$.}
\STATE $H_e \leftarrow |C^s|$.
\ENDIF
\STATE $G'_e \leftarrow H_e \mathcal{C}_e$. \COMMENT{Estimate of total number of phase 1 node joins.}
\STATE $M'_e \leftarrow \mu_b G'_e / (p_n \alpha_1)$. \COMMENT{Estimate of network size.} 
\STATE Output $M'_e$.
\end{algorithmic}
\end{algorithm}


\subsection{Resetting parameters} \label{subsec:resetparam}
Once a peer gets a good estimate of the network size after phase 1, it can (locally) determine the overlay parameters for the next b-epoch. $\mathit{ch\_dim}$ is a parameter that controls whether a dimension change is to be triggered; if it is, then $\mathit{ch\_dim}$ specifies dimension increase or decrease, and otherwise, $\mathit{ch\_dim}$ specifies no change. (Section \ref{subsec:dimchange} provides details on how $\mathit{ch\_dim}$ is determined.) $\mathcal{C}_{e+1}$ is changed to $\lambda_s\mathcal{C}_{}$ if the dimension is to be increased in the next b-epoch, or to $\mathcal{C}_{e} / \lambda_s$ if the dimension is to be decreased in the next b-epoch, where $\lambda_{s}$ is a large constant; otherwise, remains the same as $\mathcal{C}_{e}$. This helps the new nodes joining in the b-epoch $e$, to know the committee-directory mapping $\mathcal{M}_e$ for \texttt{JOIN} protocol.


When a peer mines a block after phase 1 of b-epoch $e$, in the next $\lambda_{\mathit{lb}}\log N = \alpha_2/\mu_b \beta$ blocks, it adds the overlay parameters onto the block, where $\lambda_{\mathit{lb}}$ is a suitable constant. Due to blockchain fairness, we can ensure that the majority of those blocks belong to honest peers, which helps in reaching consensus on those values by the end of the epoch. Since $\alpha \gg \beta \log N$ (Section \ref{sec:model}), as previously mentioned, $\alpha_1 = \Theta(\alpha)$. (Note that if $\beta \log N$ is greater than or close to the half-life, then the network size can change significantly during the time required for reaching an agreement over its estimate.) 

\subsection{Changing dimensions of hypercube} \label{subsec:dimchange}
The challenge is to carry out a smooth transition to a new hypercube having a different committee-directory mapping, with sufficient number of nodes in each committee, whose information is held by the associated directory nodes. The key idea is to trigger a dimension change at the end of a b-epoch, wait for one b-epoch wherein new nodes get assigned to (random) committees of the new hypercube (while the system functions using the existing hypercubic overlay), and then switch to the new hypercube in the next b-epoch. The advantage of waiting for one b-epoch for the network to get ``reconstructed'', is that the directory nodes can simply perform their functions for both hypercubic overlays until the new hypercube is adopted, avoiding complicated entry information transfer between buckets.

There is a delay of one b-epoch in switching to the next hypercube of a different dimension, i.e., if the peers decide to change the dimension during phase 2 of a b-epoch $e$, and they wait for one b-epoch, and then adopt the next hypercube in b-epoch $e+2$. The b-epoch before which the next hypercube is actually adopted, is called a \textit{transformation} b-epoch. The decision to trigger a dimension change is taken in b-epoch $e$ assuming a worst-case change of factor of 2 to the network size over the next two epochs; the decision depends on whether there is a possibility that the network size in b-epoch $e+2$ is not in $[\mathcal{C}_{e} / \lambda_s, \lambda_s\mathcal{C}_e]$ using the estimate in b-epoch $e$. We now explain the functions of directory, new nodes and existing nodes during a dimension change.



\noindent\textbf{Directory.} For each (possible) dimension, the committee-directory mapping is predetermined such that no two buckets in the directory are responsible for the same committee, and that each bucket is responsible for (almost) the same number of committees (cf. Section \ref{SUBSECdir}). We describe the behaviour of different buckets during a dimension change.

\begin{itemize}
    \item All the buckets formed in a transformation b-epoch $e$, including the middle-aged buckets at the start of the b-epoch $e$, are said to be in a \emph{split} state because they serve committees in two committee-directory mappings $\mathcal{M}_e$ and $\mathcal{M}_{e+1}$ for the current hypercube and the next hypercube, which have $\mathcal{C}_e$ and $\mathcal{C}_{e+1}$ committees respectively. In other words, the directory nodes in those buckets, run \texttt{DIR} protocol for both hypercubes simultaneously.
    \item The buckets formed after the transformation b-epoch $e$, only serve the next hypercube. 
    \item The buckets formed before the transformation b-epoch $e$ that are \emph{not} in split state (i.e., veteran buckets at the start of transformation b-epoch $e$), stop functioning from b-epoch $e+2$. 
\end{itemize}

\noindent\textbf{New nodes.} The new nodes that join the network during the transformation b-epoch $e$, get a node in both the hypercubes by considering the leftmost $\log \mathcal{C}_e$ and $\log \mathcal{C}_{e+1}$ bits of $P_{\mathit{join}}$. But they only (temporarily) operate in the old hypercube (with old mapping $\mathcal{M}_e$) in the transformation b-epoch $e$. Then, from the next b-epoch $e+1$, that node (in the old hypercube) would be considered invalid, and they start operating using the node in the next hypercube.

The new nodes contact the split state buckets and the buckets formed before them, for the entry information of nodes that joined before the transformation b-epoch $e$ using the old committee-directory mapping $\mathcal{M}_e$. And they contact the split state buckets and the buckets that are formed after them, for registering, and getting entry information about nodes that joined in or after the transformation b-epoch $e$, using the new committee-directory mapping $\mathcal{M}_{e+1}$.

\noindent\textbf{Old nodes.} All non-directory nodes that joined the network before the transformation epoch $e$ are considered invalid from the start of b-epoch $e+2$. (We resorted to simplicity by not dealing with moving the old nodes to the next hypercube, as this is sufficient for us to ensure good network connectivity. In fact, this also helps us show that switching to a new topology can be securely done during catastrophic failures. See Section \ref{subsec:rec-ana}.)

\pg{Handle $\Delta$-synchrony.} The above dimension change algorithm works well if all the peers are fully synchronized. Since some peers may reach the end of a b-epoch sooner than the others, we need to ensure that our overlay structure is maintained during a transition to the new topology. Thus, $\mu_s$ blocks are added to phase 2 (that includes the $\lambda_{\mathit{lb}}\log N$ blocks for agreement on overlay parameters). A peer functions both in old and new hypercube (through its nodes) during those $\mu_b$ blocks. Finally, after the end of b-epoch, the peer completely shifts to the new hypercube by ceasing to use its nodes in the old hypercube.



\section{Recovery}\label{sec:recovery}

Since blockchain protocols are executed indefinitely, there could be exceptional scenarios where the network may fail to provide the proven robustness guarantees simply due to very bad luck. These scenarios must be especially addressed for networks that provide probabilistic robustness guarantees. Moreover, open networks can be vulnerable to denial-of-service (DoS) attacks where a targeted set of (honest) peers instantly stop functioning, for e.g., Gnutella file sharing system, while resilient to random failures, could be split into large number of disconnected components after a targeted attack \cite{saroiu2001measurement}.

\pg{Catastrophic failure.} We say that a \emph{bucket fails} if more than 1/2 fraction of honest peers in it get \emph{corrupted}. And we say that a \emph{committee fails} if it has less than, say $20 \log N$ honest peers. Moreover, if a bucket has failed, then the committees that the bucket is responsible for, are also said to have failed. Here, the corruption could be fail-stop where the honest peers stop functioning (i.e., leave), or Byzantine where the adversary starts to control the honest peers. Consider the notion of ``catastrophic failure''\footnote{Refer to Defintion \ref{def:cat-fail} (and the previous definitions) in Section \ref{subsec:rec-ana}.} where a large set of committees and/or buckets may have failed, and in total, at most a constant fraction of honest peers get corrupted, but there still exists a subgraph of honest peers of a large size and low diameter for which bootstrapping can be securely done. Such failures model exceptional scenarios that occur in practice wherein the network is split into multiple components, resulting in considerable wastage of honest peers' hash power over time.

\pg{Recovery.} Our goal is to provably show that the network \textit{recovers} from such catastrophic failures in a short period of time. Here, we naturally define recovery as the event at which the overlay retains its native properties (that the overlay originally had before the catastrophic failure). There are two basic requirements for the overlay to recover from a catastrophic failure. First, a large fraction of honest peers can run the blockchain protocol\footnote{To allow for catastrophic failures that can instantly increase the fraction of peers that are Byzantine, and isolate some existing honest peers from the large component, we assume that the blockchain provides the required guarantees if there are least $\mu_n (1-\rho) n$ honest peers and at most $\rho n$ Byzantine peers. Refer to Section \ref{sec:model}.} (ensuring the blockchain provides the same guarantees), albeit some honest peers may end up unable to participate fully  (i.e., effective honest hash power is reduced) for a brief period of time. Secondly, the introduction service is not affected by the failure, i.e., it continues to return the most updated blockchain (among all the chains).

As long as those requirements are satisfied, the blockchain continues to make progress. Our insight is that blockchain naturally provides recovery within a constant number of half-lives, due to limited lifetime of nodes and new blocks (and nodes) being generated, a new directory will eventually be formed (replacing the failed buckets) facilitating new (honest) node joins to (failed) committees. Although our design intuitively encourages recovery, proving the same is non-trivial if the network size is allowed to significantly vary, where the peers need to actively coordinate with each other (cf. Section \ref{sec:dynamic}). For instance, estimating network size and changing dimensions after a catastrophic failure need to be carefully done and analyzed.

\pg{Extent of catastrophic failures.} We can rely on the fault-tolerant properties of the underlying topology to defend against massive catastrophic failures. The $n$-node hypercube is fault-tolerant against a linear number of random (independent) node failures \cite{hastad1989fast, chlebus1996reliable}, i.e., with high probability, it has a broadcast time of $O(\log n)$. However, if a linear number of adversarial node failures are allowed, then it is known that the hypercube can be split into components having size no more than $n / \sqrt{\log n}$ \cite{hastad1989fast}. Datar \cite{datar2002butterflies} showed that in a $n$-node multihypercube\footnote{Recall that a hypercube can be visualized by treating a row in a butterfly topology as a node, a multihypercube can similarly be visualized via a multibutterfly topology. Refer to the work of Datar \cite{datar2002butterflies} for more details.} (where the node degree is only a constant factor more than the node degree in hypercube), at least $n - O(f)$ nodes can reach at least $n - O(f)$ nodes via $\log n$-sized paths, no matter which $f$ nodes are removed.

By utilizing a multi-hypercube, recovery can be ensured even if an adversary can cause failures in up to a constant fraction of buckets including the committees that they are responsible, resulting in failures amounting up to a constant fraction of honest peers. In other words, for recovery, we can allow failures of any constant fraction of committees, provided that the buckets responsible for the remaining committees, do not experience failures. The security arguments in prior work on join-leave attacks \cite{awerbuch2004group, fiat2005making, awerbuch2009towards, guerraoui2013highly, jaiyeola2018tiny}, particularly for join protocol, heavily rely on honest majority in committees. In Appendix \ref{sec:more-rec}, we highlight the inherent difficulty of such fully localized algorithms to recover from committee failures, and provide attacks by a small number of committees having malicious majority, to keep maintaining the majority over time (which can also be used to increase the number of failed committees).

\section{Lower Bound for Half-life}\label{sec:lower-bound}
In the context of dynamic overlay networks, there is always a problem of bootstrapping a new peer into the network. How does a new peer contact an existing peer within the network? We consider the implementation of bootstrapping service $\mathcal{S}$ that has two properties:
\begin{enumerate}
\item \textit{Secure.} Responds with at least one honest peer that is within the network.
\item \textit{Bandwidth-constrained.} Each peer in $\mathcal{S}$ expends only $O(\polylog(N))$ bits in any round.
\end{enumerate}

Let the total number of peers within the network is $N$ in any round. As the best case scenario, all the peers are considered to be honest. The amount of information required to uniquely represent a peer's network address is $\Omega(\log N)$ bits, i.e., there exists no encoding scheme to compress network addresses.

\noindent\textbf{Bulletin board.} All the peers are given write access to a \textit{bulletin board} that provides read access to everyone (including the public) \cite{mitzenmacher2000useful}. Each peer can write $O(\log N)$ bits to the board per write operation. However, there is a constraint on the number of write operations per unit time. An arbitrary peer is selected to write to the board at every $\beta$ rounds. More importantly, this bulletin board is the only interface through which the peers can disseminate information to the public.

\noindent\textbf{Bootstrapping service.} Since the system is dynamic, the peers must utilize the bulletin board to regularly update the peers that are responsible for new peers to join the network. The peers follow some algorithm to construct a bootstrapping service $\mathcal{S}$ using the bulletin board. For e.g., peers may write their network addresses onto the bulletin board.

\noindent\textbf{Joining the network.} There exists some join algorithm for new peers to contact the service and join the network. The minimum requirement for a new peer to join the network is to obtain a response from the bootstrapping service.

\begin{theorem}\label{theorem:lowerbound}
Any dynamic system that implements a bootstrapping service $\mathcal{S}$ using a public bulletin board, can support a half-life of only $\widetilde{\Omega}({\beta\sqrt{N}})$.
\end{theorem}
\begin{proof}
For a new peer to obtain a response from the bulletin board, it must send a message to a network address that is on that board. We note that information representing a network address is ``useful'' only if it is ``recent''. This would mean that we can restrict our attention to the network addresses in the most recent $W$ bits of the bulletin board.

No honest peer would be alive after $\Omega(\alpha \log N)$ rounds with a high probability. Recall that an honest peer fails with probability $1/2$ in any epoch. For some suitable constant $c_1$, after $c_1 \log N$ number of epochs, an honest peer would fail with a high probability. Therefore, if there is a network address that is $c_1 \alpha \log N$ rounds old, then there is a high probability that the peer controlling that network address has left. Since the system allows $O(\log N)$ bits to be written every $\beta$ rounds, we get the following upper bound on $W$,
\begin{equation*}
    W = O(\alpha \log^2 N / \beta).
\end{equation*}

There can be other information in the most recent $W$ bits; if there are more (distinct) network addresses, then the bootstrapping service can help in joining more new peers in a fixed period of time. Therefore, as the best case scenario, we assume that those $W$ bits of information consists of only network addresses.

Let $\lambda_{\mathit{jr}}$ be the highest number of join requests that a peer can handle. As the best case scenario, we assume that all the peers in the service are distinct, and each response message is equivalent to a successful join. Recall that in any $\alpha$ consecutive rounds, at least $N/2$ new peers must be able to join the network. Let $K$ be the number of distinct network addresses in the most recent $W$ bits on the board,
\begin{align*}
    \alpha K \lambda_{\mathit{jr}} &\geq N/2.
\end{align*}
If $\alpha= o\left(\sqrt{\frac{\beta N}{ \lambda_{\mathit{jr}} \log N}}\right)$, then in total, the system can handle much less than $N/2$ join requests in $\alpha$ consecutive rounds.
\end{proof}

\section{Related Work}

There has been a large body of work on robust and efficient overlay networks in various models of churn and failure, especially in the context of distributed hash table. First, we discuss early designs in different models of failures. Then, we discuss prior work on the Byzantine join-leave attack. We also survey practical partitioning attacks and network designs pertaining to blockchain P2P networks.

\pg{Early designs.} Many overlay networks were designed to be robust against random failures, where each peer (independently) has a bounded probability of being Byzantine \cite{naor2003simple, hildrum2003asymptotically, dolev2007self, johansen2015fireflies}. Fireflies \cite{dolev2007self, johansen2015fireflies} utilizes a pseudorandom mesh structure to provide each peer with a view of the entire membership under moderate churn. Peers need to actively monitor each other by gossiping ``accusation'' and ``rebuttal'' messages. However, it is unclear how Fireflies performs in highly dynamic and large-scale systems. DHTs are made robust by replication of each data item over a logarithmic number of peers, of which a majority are honest peers \cite{fiat2002censorship, saia2002dynamically, naor2003simple, hildrum2003asymptotically}. More specifically, peers are randomly mapped to $[0, 1)$ interval, forming neighbour connections under the rules of an efficient routable topology \cite{naor2003simple, naor2007novel}. A peer would be responsible for data items hashed along a segment of $\Theta(\log n / n)$ length, where $n$ is the number of peers. Queries are routed to the destination through a majority vote at each hop.

Fiat and Saia \cite{fiat2002censorship} designed a robust content addressable network against an adversary that can adaptively corrupt up to a constant fraction of peers. They utilize the butterfly topology, where each vertex is simulated by a logarithmic number of peers. They heavily exploit properties of bipartite expanders to show that all but $\epsilon n$ peers can search successfully for all but $\epsilon n$ of the data items, for a positive constant $\epsilon$. The drawback is that their topology is fixed and does not account for peers joining and leaving. It was modified \cite{saia2002dynamically} to handle adversarial deletions under a restricted form of churn, where during any period of time in which the adversary deletes $An$ peers, then at least $Bn$ new peers join the network for positive constants $B > A$, and a new peer joins via a random peer in the network. Moreover, the join algorithm can be expensive in terms of communication cost, involving a broadcast to the entire network. Leighton and Maggs \cite{leighton1989expanders} showed that if an adversary chooses $k$ switches to fail in an $n$-input multibutterly \cite{upfal1992log}, then there will be at least $n - O(k)$ inputs and $n - O(k)$ outputs connected by $\log n$-length paths. Datar \cite{datar2002butterflies} observed that the fault-tolerance of multibutterfly networks can be used to show that a $n$-node multihypercubic network can withstand a linear number of adversarial failures. For instance, it can be shown that at least $n - 3f/2$ nodes can reach at least $n - 3f/2$ nodes via $\log n$-length paths, where an adversary can remove up to $f$ nodes. Datar \cite{datar2002butterflies} built a content addressable network, having adversarial deletion fault-tolerant guarantees similar to \cite{fiat2002censorship}, improving on the communication and storage costs, where a query requires $O(\log n)$ messages and the data is replicated in $O(1)$ nodes. However, Datar's design is not resilient against Byzantine failures.



\pg{Join-leave attacks.} Replication of a data item over a logarithmic number of peers is only helpful if there is a majority of honest peers among them. This is because the integrity of the data item is verified by considering the majority of the responses. Thus, an adversary that can overwhelm a particular region of the $[0,1)$ interval with many peers, can potentially cause harm to the system, for e.g., by transmitting false versions of a data item. In a dynamic system, where peers can join and leave, even if newly joining peers are placed in a random location in the $[0, 1)$ interval, with just $O(n)$ join (and leave) attempts, Byzantine peers can occupy a particular $\Theta(\log n / n)$ length with high probability, obtaining the majority of peers in that region (cf. Lemma 2.1 in \cite{scheideler2005spread}).

Consequently, Awerbuch and Scheideler \cite{awerbuch2004group} considered a class of attacks called join-leave attacks, where Byzantine peers collectively try to populate specific regions of the overlay topology. In this line of work, the main goal is to design a combination of join algorithm and ``network perturbation'' mechanism to continuously redistribute the Byzantine peers as new peers join the network. Here, by network perturbation, we mean that for each new peer join, some (existing) peers are placed in new (typically random) locations in the network. More formally \cite{awerbuch2009towards}, for every interval $I \subseteq [0, 1)$ of size at least $(c \log n) / n$ for a constant $c > 0$ and any polynomial number of join/leave events in $n$, the following two conditions need to be met:
\begin{itemize}
    \item Balancing condition: $I$ contains $\Theta(|I|\cdot n)$ nodes.
    \item Majority condition: honest nodes in $I$ are in the majority.
\end{itemize}
Of course, all the honest nodes in a segment of size $\Theta(\log n / n)$ could be asked to leave, in which case, the majority condition cannot be met. Therefore, the churn adversary is required to specify the join/leave sequence $\sigma$ for honest nodes in advance. But it can choose to adaptively join/leave a Byzantine node. In particular, after the first $i$ events in $\sigma$ are executed, the churn adversary can either choose to join/leave a Byzantine node or initiate the $(i+1)^{\mathrm{th}}$ event in $\sigma$. (In the beginning, all nodes are randomly placed in $[0, 1)$.)

It is easy to see that if for each join, if \emph{all} the peers are placed in random locations, then the Byzantine peers always remain well-distributed in the network. But that would be quite expensive in terms of communication cost. Therefore, ``small'' perturbations per join, typically of $\polylog(n)$ peers, is reasonable. Translating this model in the context of eclipse attacks, our goal is to design an efficient network for blockchains, formally analyze the join algorithm, and provide theoretical guarantees on connectivity over large (typically polynomial) number of join/leave events.

The key performance metrics, so far, have been join latency, join communication cost, and the capability to handle polynomial variation in network size over time. In this work, we also introduce a new metric called recovery from catastrophic failures (cf. Section \ref{sec:recovery}). Refer to Table \ref{tab:comparison} for a comparison of OverChain with prior work.

Awerbuch and Scheideler \cite{awerbuch2004group} utilize the Chord topology \cite{stoica2001chord} with each peer simulating at most $O(\log n)$ nodes at any time. They assume that at most a $O(1/\log n)$ fraction of peers can be adversarial at any time. Each node is connected to every other node in a ``region`` of size $\Theta(\log n / n)$ around it, where every region is ensured to have $\Theta(\log n)$ nodes\footnote{As previously mentioned, such $\Theta(\log n)$-sized groups of nodes act as the functional units of the network, cancelling out effects of Byzantine peers.}. The join algorithm, places a new peer (node) in a random location via a distributed random number generation (RNG) protocol run by nodes within a region (initiated by an honest node) and secure routing of the new peer to the appropriate location (through honest majority in regions). For network perturbation, they employ a limited lifetime of $O(\log n)$ rounds for each node. This work is a starting point for OverChain though our join algorithm is completely different. (We explain the choice of limited lifetime for network perturbation at a later stage.) Although they can withstand a high join rate of $O(n / \log n)$, they make a strong assumption of the existence of ``trusted gateway peers'' that can initiate unlimited join protocols at any time. If a peer needs to leave the network, then it needs to wait until all its nodes' lifetimes have expired, i.e., for $\Omega(\log n)$ rounds, in which time the entire network gets reconstructed (as in our work, where the network gets reconstructed every $\Theta(\alpha)$ rounds) due to node lifetime being set to $O(\log n)$ rounds. Using this design, the authors show that the network provides robustness (i.e., meets both the aforementioned conditions) for a polynomial number of join/leave events with high probability.

\textit{Bootstrapping and churn rate in prior work.} Awerbuch and Scheideler \cite{awerbuch2004group} are able to endure a high join rate due to the assumption of gateway peers that can help place a new peer (node) in a random location. But those gateway peers can experience churn and have limited bandwidth, leading to the lower bound proved in Section \ref{sec:lower-bound}. Subsequent works in this model, analyze a single join/leave event and provide guarantees on a series of join and leave events, all occurring one after another, for ease of exposition. Secure bootstrapping is somewhat swept under the hood. They can handle concurrent join and leave events too. And much like \cite{awerbuch2004group}, by making assumptions that provide load-balancing of join events across the network (such as existence of gateway peers), they can also handle $n/\polylog(n)$ join rate because the join algorithm (including network perturbation) takes $\polylog(n)$ work (in terms of latency and communication cost for peers involved in a join event). In this work, we go one step further and make a \emph{weaker} bootstrapping assumption of the existence of a reasonably-updated blockchain (instead of access to random or trusted peers), and provide bounds on the churn (join) rate for the system.

\textit{Catastrophic failures in prior work.} The work of Awerbuch and Scheideler \cite{awerbuch2004group} and subsequent works heavily rely on honest majority in regions\footnote{They are also referred to as quorums \cite{awerbuch2004group, awerbuch2009towards} or swarms \cite{fiat2005making} or clusters \cite{guerraoui2013highly} in the literature.} for join/leave algorithms. If there is a large adversarial attack (as in Section \ref{sec:recovery}), then the Byzantine peers easily can control majority in most of the regions over time. In Section \ref{sec:more-rec}, we argue that the existing class of algorithms fail to provide recovery guarantees for catastrophic failures. 

Subsequent works assume that at most a constant fraction of the peers are malicious. Fiat, Saia and Young \cite{fiat2005making} also use Chord \cite{stoica2001chord} as the underlying topology combined with similar logarithmic redundancy for peers. They employ the k-rotation strategy \cite{scheideler2005spread} (for k = 3) for perturbing the network during a join event. A random location $r \in [0, 1)$ is first chosen (by the region contacted by the new peer) using similar distributed RNG techniques. Two peers $p_1$ and $p_2$ are selected randomly using the algorithm in \cite{king2004choosing}. Then, the joining peer, $p_1$ and $p_2$ are ``rotated'' in that the joining peer takes the position of $p_1$, $p_1$ takes the position of $p_2$ and $p_2$ takes the position $r$. They are able to provide robustness for a linear number of join and leave events with high probability, as there exist adversarial strategies (against the k-rotation strategy) to populate a $\Theta(\log n / n)$ segment \cite{awerbuch2009towards}. However, they avoid the limited lifetime method, which puts the network in a hyperactive state where peers have to (continuously) rejoin.

In a seminal work, Awerbuch and Scheideler \cite{awerbuch2009towards} introduced the cuckoo rule for network perturbation, that provides robustness for a polynomial number of join/leave events with high probability, without the need for limited lifetime for peers. As in previous works, a random location $r \in [0, 1)$ is chosen (by the region contacted by the new peer) using similar distributed RNG techniques. Then, the new peer is placed in $r$ and all the peers situated in the segment of size $k/n$ around $r$, where $k$ is an appropriate constant (depending on the constant $\rho$, where the number of Byzantine peers are at most $\rho n$), are relocated to points in $[0,1)$ chosen uniformly and independently at random. They use the dynamic de Bruijn graph as the underlying topology. In the works discussed so far, the network size was assumed to change by at most a constant factor. While we believe that these algorithms can be augmented to provide robustness against polynomial variation in network size, it does not seem trivial to do so. (For instance, global coordination, such as network-wide agreement, does seem inevitable for a topology that is able to adapt to both polynomial network size variation and catastrophic failures.)

Guerraoui, Huc and Kermarrec \cite{guerraoui2013highly} utilize a random graph drawn from Erd\H{o}s-R\'enyi model for the underlying topology, and employ a technique known as OVER (Over-Valued Erd\H{o}s R\'enyi graph) for handling a sequence of vertex addition and deletions polynomial in $n$ whp. The underlying graph is ensured to have small degree and good expansion. Each vertex is simulated by a cluster of $O(\log n)$ size containing more than two thirds of honest peers. They heavily rely on random walks (that mix in $\polylog(n)$ rounds) for their join algorithm and network perturbation. In this work, a random cluster $C_r$ is first chosen (by the cluster contacted by the new peer) using a distributed RNG algorithm. The new peer is introduced to the peers in $C_r$ and the neighbouring clusters. Then, (for network perturbation) $C_r$ exchanges all its peers with peers chosen at random from other clusters. The updated set of peers of $C_r$ is then conveyed to the neighbouring clusters (through majority rule) by the old peers. The key contribution of this work is that they handle polynomial number of join/leave events when the network size is allowed to polynomially vary. They achieve it by splitting and merging clusters based on thresholds on cluster size. For example, if the cluster size is more than $k \log n$ for some constant $k$, then the cluster partitions its set of peers into two, informs the updated set of peers to the neighbouring clusters, and then ``adds'' the other partition to the network (using vertex addition of OVER). (Note that the security arguments heavily rely on honest majority in clusters.) Whenever a peer leaves, a similar process is carried out by the corresponding cluster. The drawback is that these algorithms require a much higher communication and round complexity than the other works. (See Table \ref{tab:comparison} for the comparison.)

\textit{Why limited lifetime method of network perturbation in OverChain?} Let us say that there is a catastrophic failure caused by a ``DoS adversary'' in the overlay. The network becomes split into multiple components, but there exists one component consisting of $\Omega(n)$ honest peers with diameter $O(\log n)$. Observe that if the network experiences little-to-no churn at this stage, even if the blockchain progresses over time, there will be very less (new) peers joining the network (due to low join rate). This is definitely possible in the model because each honest peer is assumed to fail independently with probability at most $p_{\mathit{f}}$. (More precisely, half-life $\alpha$ definition is worst-case in that there could just be one period of $\alpha$ rounds in a long period of time, say a couple of years, that $1/2$ fraction of peers join/leave the network.) And there could be no new peers joining the network for a period of time.

In this scenario, the honest peers stuck in the smaller components keep wasting their resources because their blocks may never be propagated to that large component (whose chain is essentially the main chain). Note that catastrophic failure can stealthily happen in the real world, in which case the peers may not know that the network is actually split. This is problematic because after a sufficient period of time, the DoS adversary may regain the resources to cause another large-scale adversarial attack (before the network is fully rebuilt)! Thus, compared to cuckoo rule or local splitting/merging method of network perturbation, the limited lifetime method forces rejoins of honest nodes, enabling recovery of the entire network in a constant number of half-lives. Of course, this comes at a cost of keeping the network in a hyperactive state if there is small churn on average. We argue that open peer-to-peer systems are inherently dynamic in nature, i.e., if the \textit{average} half-life is close to $\alpha$, then the system is indeed in a natural state (and not a hyperactive state), in which case, the limited lifetime method would be (asymptotically) comparable to the other methods of network perturbation.

\textit{Do redundancy and robustness go hand in hand?} The works discussed so far, have retained some form of logarithmic redundancy in terms of data integrity or secure routing to thwart the effects of malicious peers and churn. Jaiyeola et al. \cite{jaiyeola2018tiny} address the robustness guarantees that can provided with $O(\log \log n)$ redundancy. A peer is associated with a ``group'' in the interval $[0,1)$, similar to regions in \cite{awerbuch2004group}. Using group sizes of $O(\log \log n)$, they show that all but an $O(1/\polylog(n))$-fraction of groups have honest majority, and that all but an $O(1/\polylog(n))$-fraction of peers can successfully search for all but an $O(1/\polylog(n))$-fraction of data items over a polynomial number of join/leave events with high probability. A drawback is that they have a rather complicated join algorithm, where the network is reconstructed every ``epoch'', where the epoch length is set according to the departure rate of peers (as in our work). They make a strong bootstrapping assumption of the existence of ``bootstrapping groups''. Since the group size is $O(\log \log n)$, there exists $O(1/\polylog(n))$-fraction of groups, termed as ``bad'', that do not have the required number of honest peers. To avoid increase in the number of bad groups during network reconstruction, they place the peers in two networks (that exist in tandem) in two (separate) $[0, 1)$ intervals (where each of them is based on an efficient topology). We remark that their approach for network reconstruction is rather complicated; they heavily rely on the bootstrapping groups for placement of a (new) peer, creating neighbour links, and moreover, updating links as (more) peers join and become a better match as a neighbour (for e.g., they can lie closer to the neighbour points determined by the underlying topology). Similar to our work, the authors use proof-of-work to defend against Sybil attacks. (Note that the works discussed so far, in the join-leave model, made explicit assumptions regarding the number of malicious peers.) The caveat is that their network needs to continually generate (and agree on) a global random string as a (partial) input to the proof-of-work puzzles. Otherwise, the adversary can launch a pre-computation attack generating too many identities, overwhelming all the groups. Our reliance on the most recent block for proof-of-work puzzles is a significantly simpler approach than a network-wide distributed RNG algorithm.

\pg{Partitioning attacks.} Heilman et al. \cite{heilman2015eclipse} first demonstrated eclipse attacks on the Bitcoin network. Their idea is to form many incoming connections with the victim node from attacker-controlled addresses (from diverse IP address ranges), propagate many irrelevant (i.e., not part of the Bitcoin network) addresses, and wait for the victim node to restart. Once the victim node restarts, it would likely form all its outgoing connections with attacker-controlled addresses. Marcus et al. \cite{cryptoeprint:2018:236} carried out a similar attack on Ethereum's P2P network with significantly less resources by exploiting the fact that the public key of a node was its node identifier. In other words, many node identifiers could be run with just a single machine (with the same IP address). This made it easy to suitably generate and store attacker-controlled addresses in the victim node, so that after restarting, it would form all its outgoing connections with the attacker. In both papers, the authors suggested countermeasures involving the process of storing network addresses and connecting to new peers, to increase the cost of such attacks. Saad et al. \cite{saad2019partitioning} outline partitioning attacks by different adversaries such has an AS/ISP that can route Bitcoin traffic away from a target AS by BGP hijacking, a malicious mining pool that can exploit knowledge about weakly synchronized nodes to fork the network, and a software developer capable of exploiting bugs in Bitcoin client can aid other partitioning attacks.

While the previously mentioned works aimed to capture victim connections by sending many attacker-controlled (and bogus) addresses, they don't directly exploit churn or synchronization of honest peers. Saad et al. \cite{saad2021revisiting} identified that Bitcoin suffered from weak network synchronization; in a block interval of 10 mins, on average, only $\approx$39\% of the nodes had an up-to-date blockchain. The authors identify all mining nodes, which also suffered from varying network reachability. They describe a partitioning attack with a 26\% hash power adversary, by selectively broadcasting blocks to disjoint groups of miners after carefully exploiting the block propagation patterns\footnote{This type of partitioning attack is possible because Bitcoin (by default) uses a local tie-breaking rule for equal-length chains, i.e., a miner chooses to mine on the chain (or block) that was received first.}. On the contrary, Baek et al. \cite{baek2021onthe} highlight the limitations of their network monitoring systems, and assert that the block propagation in Bitcoin is indeed fast. Nonetheless, the work of Saad et al. \cite{saad2021revisiting} shows that weak synchronization can be exploited to create partitions. Our aim is to design a network that would allow for fast block propagation with at most logarithmic hops, thereby strengthening network synchronization and also paving the way for decreasing the average block interval.

Saad et al. \cite{saad2021syncattack} primarily exploit churn in Bitcoin network to create a partition between existing and newly arriving nodes. First, the adversary occupies all the incoming connections of existing nodes, which is done by exploiting the node eviction policy that favors nodes with longer connection times. Then, as nodes depart and new nodes join, they only connect with adversarial nodes from the sample of nodes provided by the DNS seeds due to unavailability of connection slots in other nodes. This gradually creates a partition between existing and newly arriving nodes. The authors suggest several countermeasures to increase the cost of such attacks such as introducing a fork resolution mechanism, restricting the number of connections from the same IP address, improving the eviction policy, etc. We argue that a theoretical framework that appropriately captures churn and malicious behaviour, is needed to analyze the security against partitioning attacks. There have also been practical attacks that consider strong network infrastructure adversaries \cite{apostolaki2017hijacking, tran2020stealthier}, which are beyond the scope of our model.

\pg{Blockchain network designs.} In light of these issues, there have been new network design proposals for blockchains. Kadcast \cite{rohrer2019kadcast} further builds on Kademlia \cite{rohrer2019kadcast} and proposes a structured broadcast protocol for disseminating blocks with at most a logarithmic number of hops and constant overhead in congestion. It is unclear how this protocol performs with respect to continuous node churn and change in network size. In Perigree \cite{mao2020perigee}, a peer retains the ``best'' subset of neighbours after regular intervals, and also continuously connects to a small set of random peers to explore potentially better-connected peers. But Perigree may actually be more prone to eclipse attacks because the adversary can easily monopolize victim peer's connections by providing well-connected peers.

\section{Discussion}

We have shown that the maintenance of robust overlay networks, in the context blockchains, can be made simpler and more efficient than fully localized algorithms. Here, we address some clarifications, limitations and natural extensions to this work.

\pg{Implementation.} As a first step, we show the theoretical results as clearly as possible, while making minimal assumptions on the blockchain and providing a clear comparison to existing algorithms and their guarantees in our model. We also include the average block interval, $\beta$, in the discussion, in addition to the interplay between churn rate, communication cost and number of Byzantine peers. We believe that a proper implementation with extensive experiments on the parameters, and a comparison with existing implementations, including the unstructured Bitcoin network, is a separate project on its own. 

\pg{Incentives.} Our work considers the honest vs Byzantine model of peers to deal with worst-case failures. But this model omits the analysis of \textit{rational} behaviour of peers. The peers currently get reward for mining blocks through block rewards and transaction fees. But for any peer to participate as a directory node and incorporate new peers into the network, it needs to expend a certain amount of communication bandwidth. If there is no incentive to do so, the peer may not follow the protocols. A key observation of our algorithm is that the communication cost is proportional to the amount of hash power owned by a peer. Thus, the next step is to design incentive-compatible overlay maintenance algorithms that are robust to churn and Byzantine failures. Such a design would promote independent peers to mine and participate in overlay maintenance.

\pg{Bootstrapping assumptions.} We make an assumption that the blockchain held by any honest peer (which can be outdated by a constant number of blocks compared to a fully updated blockchain) within the network is publicly available. This abstracts out the numerous public blockchain explorers and tie-breaking criteria (that are blockchain-specific) to consider the best among them. Note that this is \textit{not} a solution to the bootstrapping problem of how a new peer can find existing peers, as some peers within the network must maintain such blockchain explorers, thus acting as an interface for entities outside and inside the network.

We claim that our bootstrapping assumption is \textit{weaker} than assumptions made by current (practical and theoretical) designs. Currently, the bootstrapping process is highly centralized (and static), with only a handful of peers responsible for providing information about random peers\footnote{In unstructured networks such as Bitcoin, it is important that new peers get connected to random peers for the network to simulate a random graph that has desirable connectivity properties.} to new peers. More importantly, there is no way of verifying the randomness of given peers. On the other hand, we utilize the properties of blockchain such as safety, liveness and fairness\footnote{Blockchain protocols do not provide strong guarantees such as a block is produced by a random peer. In that case, having a network address in each block trivially provides access to random peers.} to ensure that the overlay structure is (provably) maintained.

\pg{Join latency.} Our algorithm allows a peer to join the network in constant number of rounds, in contrast to previous algorithms that required logarithmic number of rounds. It is important for us to dissociate the time taken by a peer to mine a new node and the time taken for a node to join the network. Note that the node mining is only done to regulate the join rate of peers, to prevent a Sybil attack. Prior works \cite{awerbuch2004group, fiat2005making, awerbuch2009towards, guerraoui2013highly} subsume such a Sybil mechanism and makes explicit assumptions on the number of peers joining in a given period of time. In other words, we could remove the node mining aspect, and just use the hash function for associating a node with a random committee. Thus, the join latency, a key performance metric in overlay maintenance, is the time taken for a peer to join a committee and obtain information about all its neighbouring peers.

\pg{Heterogeneous peers.} In our work, a peer generates nodes (to facilitate overlay maintenance) and blocks (which are also tied to overlay maintenance via directory nodes). In other words, a peer is also (necessarily) a miner. (This design is suitable for proof-of-work blockchains because 2-for-1 PoW mining \cite{garay2015bitcoin, pass2017fruitchains} can be used for simultaneously mining for both blocks and nodes.) But in practice, peers part of the blockchain P2P network, can be \emph{signficantly} different from each other in terms of (1) half-life (for e.g., there could be a few long-lived honest peers which can be exploited to make the overlay protocols efficient), (2) hash power (for e.g., non-miners vs lone miners vs mining pools), (3) blockchain verification (for e.g., full nodes vs SPV clients in Bitcoin), (4) bandwidth (for e.g., a few peers can endure high communication cost), etc. Thus, an important research direction is to come up with a theoretical framework that models the peer heterogeneity, and then design robust overlay maintenance algorithms.


\pg{Beyond PoW for Sybil defense.} There are two components of OverChain: (1) Sybil defense mechanism to regulate the number of new nodes generated per round, (2) and overlay maintenance features such as placing new nodes in random committees (using the output of a hash function), handling join requests, broadcasting network information (including blocks), etc. The latter component can be translated onto any type of blockchain without any changes. The challenge is to design an efficient Sybil defense mechanism (for node mining), possibly by making use of the same restricted resource (of the blockchain).

\section{Full Analysis}\label{sec:full-ana}
The overlay network has certain desirable properties that are subsequently proven in this section. The blockchain properties are important for the proofs. Blockchain fairness is used for showing a bound on the communication cost due to \texttt{JOIN} protocol, bound on the number of honest nodes in any bucket of the active directory, and reaching an agreement on the overlay parameters via honest majority in the blocks of phase 2 of a b-epoch. Both blockchain safety and liveness are implicitly required, for e.g., if the blockchain provided by the introductory service is different from the chains within the network, then the \texttt{JOIN} protocol would not work. Besides that, blockchain liveness is explicitly used for obtaining the (approximate) time taken for adding a large set of consecutive blocks onto the confirmed chain. Although a sequence of b-epochs is well-defined with respect to a peer (and its confirmed chain), we divide the time into b-epochs (for the system as a whole) even for the analysis. We consider that a b-epoch ends if any of the honest peer reaches the end of the b-epoch.

\begin{table}[h!]
\begin{center}
\begin{tabular}{|c|l|} 
 \hline
 \multicolumn{1}{|c|}{Symbols} & \multicolumn{1}{c|}{Meaning} \\ 
 \hline
 $N$ & Maximum network size \\ 
  \hline
 $\rho$ & Fraction of total number of peers that are malicious \\ 
 \hline
 $\alpha$ & Number of rounds in one epoch \\ 
 \hline
 $\mathcal{B}$ & Number of buckets in one directory \\
 \hline
 $\mathcal{K}$ & Number of blocks in one directory \\ 
 \hline
 $\mathcal{B}_{\mathit{act}}$ & Number of buckets in the active directory \\
 \hline
 $\mathcal{K}_{\mathit{act}}$ & Number of buckets in the active directory \\ 
 \hline
 $T_{l}$ & Lifetime of a non-directory node (in terms of number of blocks) \\
 \hline
 $T_{\mathit{dl}}$ & Lifetime of a directory node (in terms of number of blocks) \\
 \hline
 $\lambda_{\mathit{jr}}$ & Highest number of join requests handled by a directory node per round\\
 \hline
 $\mathcal{C}_e$ & Number of committees in b-epoch $e$ \\ 
 \hline
 $L'_e$ & Estimate of the network size at the end of b-epoch $e$ that is agreed\\& upon by all the honest peers \\
 \hline
 $L_e$ & Network size at the end of b-epoch $e$ \\ 
 \hline
 $M'_e$ & Estimate of the network size after phase 1 of b-epoch $e$ that is agreed\\& upon by all the honest peers\\
 \hline
 $M^L_e$ & Minimum network size in phase 1 of b-epoch $e$ \\ 
 \hline
 $M^H_e$ & Maximum network size in phase 1 of b-epoch $e$ \\ 
 \hline
 $G'_e$ & Estimate of the number of nodes that joined in phase 1 of b-epoch $e$ \\ 
 \hline
 $G_e$ & Number of nodes that joined in phase 1 of b-epoch $e$ \\
 \hline
 $\delta_{\mathit{err}}$ & Error parameter (that is less than 1) for network size estimation \\ 
 \hline
\end{tabular}
\caption{Important symbols and their meaning.}
\label{tab:notation}
\end{center}
\end{table}

\begin{definition}\label{def:actrobust}
The active directory (bootstrapping service) is considered to be \text{\normalfont robust} if it satisfies the following properties:
\begin{enumerate}
    \item Each bucket has at least $\lambda_b \log^2 N$ honest nodes where $\lambda_b$ is some positive constant.
    \item The entry information of any honest node that is part of the network, is held by (the honest directory nodes of) the appropriate bucket (determined by the committee-directory mapping).
\end{enumerate}
\end{definition}

\begin{definition} \label{DEFjoin}
The \texttt{JOIN} protocol, for a new node $q$ generated by peer $u$, is said to be \text{\normalfont successful} if it satisfies the following properties:
\begin{enumerate}
    \item Peer $u$ must get entry information about all (honest) nodes of the committee that node $q$ is going to join, and about all (honest) nodes of the neighbouring committees.
    \item Peer $u$ sends node $q$'s entry information to all the (honest) nodes in the committee that $q$ is going to join, and to all the (honest) nodes of the neighbouring committees.
\end{enumerate}
\end{definition}

\begin{definition} \label{DEFpartres}
The overlay network is considered to be \text{\normalfont partition-resilient} if it satisfies the following properties:
\begin{enumerate}
    \item Each committee has $O(\log N)$ nodes and at least $\lambda_p \log N$ honest peers for a positive constant $\lambda_p$.
    \item In each committee, every pair of honest nodes are connected.
    \item Each honest node has $\Omega(\log N)$ honest neighbours in each of its neighbouring committee.
\end{enumerate}
\end{definition}

\begin{definition}\label{DEFbandadequate}
A b-epoch $e$ is said to be \text{\normalfont bandwidth-adequate} if each (honest) peer needs to send or receive $O(\log^3(N))$ messages for overlay maintenance in all rounds of b-epoch $e$.
\end{definition}

\begin{definition}\label{DEFgoodestratio}
Let the quantity $R_e = L_e / L'_e$ be defined as the \text{\normalfont estimate ratio} of b-epoch $e$. A b-epoch $e$ is said to have a \textit{good estimate ratio} if,
\begin{equation*}
    R_{e} = \frac{L_e}{L'_e} \in \left[\frac{1}{2\mu_b(1+\delta_{\mathit{err}})}, \frac{2 \mu_b}{(1-\rho)(1-\delta_{\mathit{err}})}\right].
\end{equation*}
\end{definition}

\begin{definition} \label{DEFstableb-epoch}
A b-epoch $e$ is said to be a \text{\normalfont stable} b-epoch if the number of peers in any round $r$ of b-epoch $e$, denoted by $N^r_e$, is $N^r_e \in [\mathcal{C}_e / \lambda_{s}, \lambda_{s} \mathcal{C}_e]$, for some positive constant $\lambda_{s}$.
\end{definition}

\begin{remark} \label{rem:exp-nodes}
The expected number of new nodes that can be generated in a stable b-epoch, is in $[(\lambda_{n} (1-\rho) \mathcal{C}_e\log N) / (\mu^2_b \lambda_s), \lambda_{n} \lambda_s \mathcal{C}_e\log N]$. This is because the node difficulty threshold, $p_n = (\lambda_n \log N) / (q \alpha)$, is fixed, the number of rounds in a b-epoch is bounded using blockchain liveness, and the network size is bounded in a stable b-epoch.
\end{remark}

\begin{definition}\label{DEFdiffrecal}
A b-epoch $e$ is said to be \text{\normalfont synchronized} if:
\begin{enumerate}
    \item All honest peers calculate $M'_e$ such that $M'_e \in [(1-\delta_{\mathit{err}})(1-\rho)M^L_e / \mu_b, (1+\delta_{\mathit{err}})\mu_b M^H_e]$ at the end of phase 1, as specified in Algorithm \ref{alg:sizeest}.
    \item $\mathcal{C}_{e+1}$ and $\mathit{ch\_dim}$ are added to each honest block confirmed after phase 1.
    \item There is an honest majority among the blocks confirmed in phase 2.
\end{enumerate}
\end{definition}

\begin{lemma}\label{lemma:totalnodes}
If the b-epochs $e, e-1, \dots, z$ where $z = \mathrm{max}(1, e-\lceil \lambda_{l} \mu_b \rceil)$, are stable b-epochs, then whp, every committee can have $O(\log N)$ nodes.
\end{lemma}
\begin{proof}
The nodes generated before the previous $\lambda_{l} \mu_b$ b-epochs are considered invalid since the start of b-epoch $e$ (due to blockchain liveness and non-directory node lifetime). Therefore, by Chernoff bounds, the total number of nodes in the system during b-epoch $e$ is at most $O(\mathcal{C}_e \log N)$ with high probability (due to Remark \ref{rem:exp-nodes} and the fact that the network size can change by at most a constant factor over a constant number of epochs). By using a balls-and-bins argument, each committee can have $O(\log N)$ nodes with high probability, for a large enough $\lambda_{n}$.
\end{proof}

\begin{lemma} \label{lemma:peer-node-comm}
If the b-epochs $e, e-1, \dots, z$ where $z = \mathrm{max}(1, e-\lceil \lambda_{l} \mu_b \rceil)$, are stable b-epochs, then each peer controls at most $O(\log N)$ non-directory nodes in the network and at most $c_m$ non-directory nodes in any committee in b-epoch $e$, whp, where $c_m > 3$ is some constant.
\end{lemma}
\begin{proof}
We examine the probability that a constant number of all nodes controlled by a peer joining a particular committee.

Recall that a node is valid for $T_l$ blocks, amounting to at most $\lambda_l \mu_b \alpha$ rounds due to blockchain liveness. The node difficulty threshold, $p_n = (\lambda_n \log N) / (q \alpha)$, is fixed. To calculate the number of nodes controlled by a peer, we need to look back for at most $\lambda_l \mu_b \alpha$ consecutive rounds before round $r$, as nodes generated before it will not be considered valid in round $r$. Applying a Chernoff bound, whp, a peer can control at most $D = O(\log N)$ nodes in any round $r$. Applying a union bound over all peers and $\alpha$ rounds, this holds for all peers during the entire b-epoch $e$.

Each of these $D$ new nodes generated by the peer get mapped to a random committee (by random oracle assumption). This can be viewed as throwing $D$ balls in $\mathcal{C}_e \geq N^{1/y} / \lambda_{s}$ bins. Let the upper bound for the number of balls in a bin be some large enough constant $c_m > 3$. Let $Z_k$ be the random variable that denotes the number of nodes in committee $k$. By Chernoff bounds and large enough $N^{1/y}$ (minimum network size),
\begin{equation*}
    \mathrm{Pr}[Z_k \geq c_m] \leq e^{c_m} \cdot \left(\frac{D}{\mathcal{C}_e}\right)^{c_m} \leq \left(\frac{1}{\mathcal{C}_e}\right)^{c_m/2},
\end{equation*}
committee $k$ has at most $c_m$ nodes with high probability. By applying a union bound on the number of committees, every committee has less than $c_m$ nodes controlled by the peer with high probability. Finally, applying a union bound over the total number of peers, each peer controls at most $c_m$ nodes in any committee.
\end{proof}

\begin{lemma} \label{lemma:hon-peer-lower}
For $e \geq 2$ and constant $\lambda_p > 0$, if b-epochs $e-1$ and $e$ are stable, then in b-epoch $e$, each committee has at least $\lambda_p \log N$ honest peers mapped to it with high probability.
\end{lemma}
\begin{proof}
Recall that the node difficulty threshold, $p_n = (\lambda_n \log N) / (q \alpha)$, is fixed. By stable epoch-property, blockchain liveness (for number of rounds in a b-epoch), and the fact that at most $\rho$ fraction are Byzantine, the expected number of new honest nodes, in b-epoch $e-1$, is at least $(\lambda_n \mathcal{C}_{e-1} (1-\rho) \log N) / (\mu^2_b \lambda_s)$. By Chernoff bounds, except with exponentially low probability, the total number of honest nodes generated is at least $\Theta(\lambda_n \mathcal{C}_{e-1} \log N)$. (We drop the other constants unless necessary, as $\lambda_n$ controls the failure probability.) These nodes get randomly mapped to $\mathcal{C}_{e}$ committees. (Note that $\mathcal{C}_e = \Theta(\mathcal{C}_{e-1})$ as b-epochs $e-1$ and $e$ are stable, and the network size can change by at most a constant factor over a constant number of epochs.) Using a balls-and-bins argument, for large enough $\lambda_n$, the number of honest nodes mapped to a committee is at least $\Theta(\lambda_n \log N)$ with high probability. By Lemma \ref{lemma:peer-node-comm}, since each honest peer controls at most $c_m$ number of nodes in a committee, the number of honest peers mapped to a committee is at least $m = \Theta(\lambda_n \log N)$ with high probability. (We drop $c_m$ in $m$ too because an increase in $\lambda_n$ does not require $c_m$ also be increased by the same factor for \ref{lemma:peer-node-comm} to hold, for large enough $N^{1/y}$.)

Every node (in its committee) survives until the end of the next b-epoch once it enters the network, whether there is a dimension change or not. (If b-epoch $e-1$ is a transformation b-epoch, then nodes get mapped to a new hypercube. If b-epoch $e$ is transformation b-epoch, then the node is considered valid until the end of b-epoch $e$.) Our aim, in this lemma, is only to bound the number of honest nodes that get mapped and survive until the end of the next b-epoch. (We are not yet considering the success of join protocol, dimension change, etc.) But the peer controlling a node can leave the network. Thus, after two b-epochs, amounting up to two half-lives (via blockchain liveness), the expected number of honest peers in a committee is at least $m^{1/4}$. Applying a Chernoff bounds on survivability of an honest peer in a committee (since their failure is independent of other honest peers), the number of honest peers mapped to that committee that survive until the end of b-epoch $e$, is $\Theta(\lambda_n \log N)$ with high probability. Finally, applying a union bound on all the committees, the lemma holds with high probability, for a large enough constant $\lambda_n$.
\end{proof}

\begin{lemma}\label{lemma:DIRbandwidthcost}
If b-epoch $e$ is a stable epoch, then each peer receives at most $O(\log^2 N)$ \text{\normalfont\texttt{JOINING}} and $O(\log^2N)$ \text{\normalfont\texttt{REQ\_INFO}} messages in any round in b-epoch $e$ with high probability.
\end{lemma}

\begin{proof}
As this proof is quite involved, we provide a high-level intuition. Let a directory node receive at most $R$ \texttt{REQ\_INFO} messages in a round. Let $T$ be the number of directory nodes controlled by an honest peer. Then, that honest peer needs to send $O(T R \log N)$ node entry information (through \texttt{COMM\_INFO} messages) if there are $O(\log N)$ nodes in each committee. As we shall subsequently prove, it turns out that $T = O(\log N)$. $R$ must be $O(\log N)$ so that the overall communication cost per round for the peer is $O(\log^3 N)$. We set $\lambda_{\mathit{jr}} = O(\log^2 N)$, the number of (valid) new nodes mapped to a bucket, so that on expectation, a directory node in that bucket needs to handle $O(\log N)$ \texttt{REQ\_INFO} messages. We focus on getting an upper bound on the join requests for any bucket per round. (We use the term ``join request'' in place of a new (valid) node.) Then, finally, we bound the number of \texttt{REQ\_INFO} messages using the bound on the join requests and the random sampling used for sending \texttt{REQ\_INFO} messages (cf. Algorithm \ref{alg:join}).

We split the proof into two cases: (1) the network size is at most $c_1 \sqrt{N} \log^2 N$ but greater than $N^{1/y}$ (minimum network size) in all rounds of b-epoch $e$, where $c_1$ is some constant, and (2) the network size is greater than $\Omega(\sqrt{N}\log^2 N)$ in any round. In case 1, the network size is comparable to the number of blocks in the directory. This means that a peer may control many blocks in the active directory, but the overall join rate (per round) itself turns out to be quite low. This is because about $\widetilde{\Theta}(\sqrt{N})$ nodes join in a b-epoch (which is $\widetilde{\Theta}(\beta\sqrt{ N})$ rounds due to blockchain liveness). In case 2, however, the join rate per round can go up to $\widetilde{\Theta}(\sqrt{N})$. But due to blockchain fairness, it turns out that a peer controls at most $O(\log N)$ blocks in the active directory, as the network size is considerably higher than the number of blocks in the active directory. This means that $\widetilde{\Theta}(\sqrt{N})$ join requests generated per round get load-balanced across the $\widetilde{\Theta}(\sqrt{N})$ buckets in the active directory, where there are sufficient peers to handle them. The tricky part is, of course, bounding the communication cost within the required log factors.

In a stable b-epoch $e$, the expected number of node join requests that can be generated per round (due to $p_n$) is at most $d = \Theta((\mathcal{C}_e \log N) / \alpha)$. Let $D_r$ be the actual number of valid join requests generated per round $r$ in b-epoch $e$.

\textit{Case 1.} A single peer may control many blocks in a directory, as the network size is comparable to the number of blocks in the active directory. However, the communication cost for an entire directory is low because the join rate (per round) is low.

The expected number of join requests per round is $d = O((\log^2 N) / \beta)$ because $\mathcal{C}_e = O(\sqrt{N} \log^2 N)$ (as the network size changes by at most a constant factor over a constant number of epochs, and a b-epoch is at most a constant number of epochs due to blockchain liveness). The adversary can choose to conduct a pre-computation attack, wherein it can send all its (valid) join requests right before the next block is confirmed. Due to blockchain liveness, the next block is confirmed in at most $O(\beta)$ rounds. In such a scenario, the expected number of join requests over a period of one block is $O(\log^2 N)$ (via linearity of expectation). In other words, the expected number of join requests in any round $r$ is $O(\log^2 N)$. This means that $D_r = O(\log^2 N)$ whp, by Chernoff bounds. Applying a union bound over all rounds in b-epoch $e$, the join rate (per round) is $O(\log^2 N)$ whp. Note that this is for an entire directory (and there are at most a constant number of directories in the active directory).

\textit{Case 2.} Due to the fairness property, for the case in which the network size is $\Omega(\sqrt{N} \log^2 N)$, we show that the number of blocks held by a \emph{single} peer in a directory is low. Let $N_{r}$ be the network size in round $r$ in b-epoch $e$. Let $x$ be the fraction of blocks generated by one peer in a segment of the chain of length $L$, then whp, by fairness property, for $\delta > 0$ and a constant $c_2$,
\begin{align*}
    x &\leq 1 - (1-\delta)\left(\frac{N_{r}-1}{N_{r}}\right)\\
      &\leq \delta + \frac{1-\delta}{N_{r}}\\
      &\leq \frac{c_2 \kappa}{L}.
\end{align*}
Since $L = \mathcal{K}_{\mathit{act}} = \Theta(\alpha / \beta) = \Theta(\sqrt{N} \log N)$ (by Section \ref{subsec:setparam}) and $N_{r} = \Omega(\sqrt{N} \log^2 N)$, the second term in the second inequality can be neglected for a large enough constant $c_2$. That peer gets at most $x L \leq c_2 \kappa = O(\log N)$ blocks. (As the worst case scenario, assume that these blocks occur in different buckets.) This needs to hold for all the peers and across all the $O(\alpha)$ rounds in the b-epoch. Applying a union bound on all the peers and rounds, each peer has at most $O(\log N)$ blocks in the active directory in any round, with a high probability.

Let us calculate an upper bound for the number of valid join requests per bucket. Let $C_b$ be the set of committees that the bucket $b$ is responsible for. We will split this analysis into two subcases. First, we will upper bound the number of join requests for the committees in $C_b$. Then, we will upper bound the number of join requests arising from the neighbouring committees of each committee in $C_b$.

Before that, let us get an upper bound on $D_r$ for any round $r$ as it is required for the analysis of the two subcases. Using the stable b-epoch property, and when the network size is $\Omega(\sqrt{N} \log^2 N)$,
\begin{align*}
    d &= \Theta\left( \frac{\mathcal{C}_e \log N}{\alpha}\right) = \Omega\left(\frac{\log N}{\beta}\right).
\end{align*}
As the worst case scenario (for obtaining an upper bound), the adversary can conduct a similar pre-computation attack as in Case 1. In such a scenario, the above expected number of join requests over a period of $O(\beta)$ rounds (via blockchain liveness), gets multiplied by a factor of $O(\beta)$ (via linearity of expectation). This means that for a large enough constant $\lambda_n$, $D_r = \Theta((\beta \mathcal{C}_e \log N) / \alpha)$ with high probability (via Chernoff bounds).

\textit{Case 2, Subcase 1.} Let the upper bound for the expected number of join requests for bucket $b$ in round $r$ be denoted as $u_r$. Recall that the network size can be at most $N$ (and $C_e \leq \lambda_s N$ due to stable b-epoch property). Applying the upper bound on the network size, and using $\lambda_{\mathit{jr}} = O(\log^2 N)$,
\begin{align*}
    u_r \leq \frac{D_r}{\mathcal{B}} &= \Theta\left( \frac{\beta \mathcal{C}_e \log N}{\alpha \mathcal{B}}\right) = \Theta\left( \frac{\beta \mathcal{C}_e \lambda_{\mathit{jr}} \log N}{\beta N \log^2 N}\right)\\
    &= O(\log N).
\end{align*}
Applying Chernoff bounds (via balls-and-bins argument), the number of \texttt{JOINING} requests for each directory node in bucket $b$ is $O(\log N)$ in a round of b-epoch $e$.

\textit{Case 2, Subcase 2.} Let the set of all neighbouring committees of each committee in $C_b$ be denoted as $N_b$. Let the upper bound for the expected number of join requests arising from committees in $N_b$ be denoted as $u^n_r$. Let us denote the number of links of a committee $c$ in $N_b$ to committees in $C_b$, as $w_c$. If a join request is sent for committee $c$, then $w_c$ requests are sent to bucket $b$ (as part of getting to know the neighbouring committees' nodes). Let the maximum number of neighbouring links for a committee in the topology be $M = \log N$. Thus, the sum of $w_c$ over all committees $c$ belonging to $N_b$ is at most $M |C_b|$.

Let $X(j)$ be the random variable that takes input $P_{\mathit{join}}$ of a valid join request, and is equal to $w_c / M$ if the join request $j$ maps the node to committee $c \in N_b$, and 0 otherwise. Note that $X(j)$ over the join requests are independent and belong to $[0, 1]$. Recall that the committee-directory mapping (in any dimension) is set such that for any bucket $b$, $|C_b| \approx \mathcal{C}_e / \mathcal{B}$. The expected value of $X(j)$ for a join request $j$ is,
\begin{align*}
    \sum_{c \in N_b} \frac{1}{C_e} \cdot \frac{w_c}{M} = \frac{M|C_b|}{C_e M} = O\left(\frac{1}{\mathcal{B}}\right).
\end{align*}
Summing over all join requests,
\begin{align*}
    \sum_{j} \mathbb{E}[X(j)] &\leq O\left(\frac{D_r}{\mathcal{B}} \right)  = O(\log N),
\end{align*}
where the last equation is from the analysis of subcase 1. Applying Chernoff bounds, the summation is $O(\log N)$ with high probability. Now, multiplying the summation by $M$, which is $O(\log^2 N)$, gives the total number of new nodes that send \texttt{REQ\_INFO} message to directory nodes in that bucket. On expectation, a directory node in that bucket receives at most $O(\log N)$ (due to random sampling in Algorithm \ref{alg:join}). And finally, by applying Chernoff bounds and union bound, each directory node receives at most $O(\log N)$ \texttt{REQ\_INFO} messages whp.

Byzantine peers may not do the random sampling. But it can be enforced by the usage of hash function (i.e., verifiable randomness). The directory nodes in the bucket are ordered by the block numbers. Thus, $(\log \log N)$ bits are necessary to represent each of them. Appending $1, 2, \dots, O(\log \log N)$ to the input that provided the valid node, and applying the hash function, provides $O(\log N \log \log N)$ random bits, that can be used to sample $O(\log N)$ directory nodes in the bucket. Note that $q$, which is the maximum number of hash queries by a peer per round, is quite large in practice. These extra $\polylog(N)$ hash queries per round are negligible, compared to the number of hash queries required to mine a node (which is set by the difficulty threshold).

Since a peer controls at most $O(\log N)$ directory nodes in the active directory in any round, then the total number of \texttt{JOINING} and \texttt{REQ\_INFO} requests for a peer is at most $O(\log^2 N)$ per round.

Applying a union bound over all peers and rounds, each peer receives at most $O(\log^2 N)$ \texttt{JOINING} messages and $O(\log^2 N)$ \texttt{REQ\_INFO} messages with high probability.
\end{proof}

\begin{theorem} \label{th:totalbandwidth}
If the b-epochs $e, e-1, \dots, z$, where $z = \mathrm{max}(1, e-\lceil \lambda_{l} \mu_b \rceil)$, are stable b-epochs, then whp, b-epoch $e$ is bandwidth-adequate.
\end{theorem}
\begin{proof}
We first show that the communication cost for the non-directory nodes and directory nodes controlled by any honest peer is $O(\log^3 N)$ with high probability. Moreover, any newly joining node also sends or receives at most $O(\log^3 N)$ messages with high probability. In other words, there is sufficient bandwidth for peers to correctly execute the overlay protocols.

By Lemma \ref{lemma:totalnodes}, each committee can have at most $O(\log N)$ nodes with high probability.

\textit{Non-directory node.} By Lemma \ref{lemma:peer-node-comm}, any honest peer controls at most $O(\log N)$ nodes in the network and at most $c_m$ nodes in any committee in b-epoch $e$. This allows an honest peer to broadcast $O(\log N)$ messages from a committee and simultaneously relay $O(1)$ messages from all its nodes (in different committees), to all neighbouring nodes. For example, during network size estimation, after phase 1, all nodes of a (random) committee broadcast the entry information of the nodes that joined the committee during phase 1. In that case, a node in that committee would be broadcasting entry information of $O(\log N)$ nodes to all the $O(\log^2 N)$ neighbouring nodes in a single round.

\textit{Directory node.} By Lemma \ref{lemma:DIRbandwidthcost}, any honest peer receives at most $O(\log^2 N)$ \texttt{JOINING} and \texttt{REQ\_INFO} messages in any round of b-epoch $e$. A directory node has to respond to a \texttt{REQ\_INFO} with a \texttt{COMM\_INFO} message consisting of entry information of nodes in a committee. Thus, the communication cost for a peer due to \texttt{COMM\_INFO} is $O(\log^3 N)$.

\textit{Joining node.} The node sends \texttt{JOINING} messages to all directory nodes in one middle-aged bucket. This amounts to $O(\log^2 N)$ messages in total, as there are $O(\log N)$ directory nodes in a bucket. The communication cost due to \texttt{REQ\_INFO} messages is the same as that of \texttt{JOINING} messages, as they are sent to just $O(\log N)$ directory nodes in $O(\log N)$ buckets. As shown above, honest peers within the network can handle all the join requests, and send back (valid) committee entry information in \texttt{COMM\_INFO} messages, which amounts to $O(\log^3 N)$ messages (by Lemma \ref{lemma:totalnodes}). Finally, the node sends \texttt{JOINING} messages to all nodes in $O(\log N)$ committees, resulting in $O(\log^2 N)$ messages (by Lemma \ref{lemma:totalnodes}).

\end{proof}

\begin{theorem} \label{the:actdir-robust}
If the b-epochs $e, e-1, \dots, z$ where $z = \mathrm{max}(1, e- \lceil\lambda_{ \mathit{dl}} \mu_b \rceil)$, are bandwidth-adequate b-epochs, then in b-epoch $e$, the active directory is robust with high probability.
\end{theorem}
\begin{proof}
For achieving property 2 of a robust active directory (Defintion \ref{def:actrobust}) over a period of b-epoch, there are two key requirements: (1) the entry information of a new node $q$ should be stored by honest nodes in the appropriate middle-aged bucket $b$, and (2) throughout the lifetime of a node $u$, its entry information is stored by honest nodes in bucket $b$. For the first requirement, the honest peers controlling those directory nodes should have enough bandwidth to receive all the join requests. This can be ensured as it is already given that b-epoch $e$ and the last $\lceil \lambda_{dl} \mu_b \rceil$ b-epochs are bandwidth-adequate. And both requirements are contingent on the fact that there are a sufficient number of honest nodes in each bucket. Therefore, we focus on proving that property 1 of a robust active directory (Defintion \ref{def:actrobust}) is achieved with high probability.

Let us calculate the probability that any honest directory node does not leave before completing the lifetime of a directory node. Since the lifetime is set to be $T_{\mathit{dl}}$ blocks, which amounts to at most $\lambda_{\mathit{dl}} \mu_b$ epochs. The probability that an honest directory node would stay, denoted by $P$, is at least $2^{-\lambda_{\mathit{dl}} \mu_b}$.


A constant fraction of the directory nodes in a bucket are honest using blockchain fairness. The number of honest blocks (or directory nodes) in a bucket, denoted by $M$, is at least $(1-\rho)(1-\delta)\lambda_{d} \log^2 N$, with high probability, for some $\delta \leq O(1/\log N)$.

Again, due to blockchain fairness, any peer controls at most $O(\log N)$ directory nodes in a bucket. Let us say $x$ is the fraction of blocks generated by one peer in one bucket, then whp by fairness property, for a small $\delta' > 0$, and where $L = \lambda_d \log^2 N$ is the number of blocks in a bucket and $N_{r}$ is the network size in round $r$ in b-epoch $e$,
\begin{align*}
    x &\leq 1 - (1-\delta)\left(\frac{N_{r}-1}{N_{r}}\right)\\
      &\leq \delta' + \frac{1-\delta}{N_{r}}\\
      &\leq c\frac{\kappa}{L},
\end{align*}
for some suitable constant $c > 0$. Note that the second term in the second inequality is negligible, as $N_r \geq N^{1/y}$ in any round. Thus, the number of blocks generated is at most $x L \leq c \kappa = c_1 \log N$, where $c_1 > 0$ is a constant, and $\kappa = O(\log N)$. Applying a union bound on all the peers, each peer has at most $m = O(\log N)$ blocks in a bucket.

Let us now show that there are always at least $\lambda_{b} \log^2 N$ honest directory nodes in a bucket. Let $N(p)$ be the number of blocks generated by an honest peer $p$ in a certain bucket. (Note that this is fixed once the bucket is formed.) Let $X(p)$ be the random variable that is equal to $N(p) / m$ if $p$ stays for the entire directory node lifetime, otherwise equal to 0. The expected value of the summation over all the honest peers is $\sum_{p} (P N(p)) / m = (P M) / m$. For large enough $\lambda_{d}$, applying a Chernoff bound, the value of the summation is $\Theta((P M) / m)$ with high probability. (Note that if $\lambda_d$ is increased, then the size of the bucket increases, but this can be compensated by a constant factor decrease in the number of buckets leading to a constant factor increase in the bandwidth cost.) Multiplying the summation by $m$, gives the total number of honest nodes that stay for the entire directory node lifetime, which is $\lambda_b \log^2 N$.

Applying a union bound on the $\mathcal{B}_{\mathit{act}}$ buckets at the start of b-epoch $e$ and the buckets generate in b-epoch $e$, this holds for all the buckets in the active directory in b-epoch $e$. Thus, property 1 holds with high probability. As discussed earlier, property 2 automatically follows if property 1 is satisfied for all the required buckets.
\end{proof}

\begin{lemma}\label{lemma:join-succ}
If a b-epoch is bandwidth-adequate and the active directory is robust in that b-epoch, then the \text{\normalfont\texttt{JOIN}} protocol is successful with high probability.
\end{lemma}
\begin{proof}
Recall that an implication of a b-epoch being bandwidth-adequate is that the directory nodes can handle, i.e., receive and respond to all join requests sent to them. Moreover, since the active directory is robust, the (honest) directory nodes in each bucket have the correct membership information.

Let us first focus on any one bucket that the new node sends to join request to. The probabilistic guarantee of the success of \texttt{JOIN} protocol because of the random directory node sampling done for \texttt{REQ\_INFO} messages. If an honest directory node in that bucket receives \texttt{REQ\_INFO}, then it sends the required entry information to the new node. This is sufficient because the Byzantine directory nodes can only under-represent the committee nodes. 

Thus, we calculate the probability that the new node samples an honest directory node in that bucket. Since it picks $\lambda_j \log N$ uniform and independent samples, for large enough $\lambda_j$, the probability that none of them honest nodes is $\left[1 - \frac{\lambda_b}{\lambda_d}\right]^{\lambda_j \log N} \leq \frac{1}{N^{c}}$, for some constant $c > 2$.

Applying a union bound over all the $O(\log N)$ buckets that the new node contacts, the \texttt{JOIN} protocol is successful with a high probability.
\end{proof}

\begin{theorem} \label{the:partres}
If the b-epochs $e, e-1, \dots, z$ where $z = \mathrm{max}(1, e-\lceil \lambda_{l} \mu_b \rceil)$, are stable and bandwidth-adequate, and if the active directory is robust in those b-epochs, then whp, the overlay network is partition-resilient in b-epoch $e$.
\end{theorem}
\begin{proof}
Since the last $\lambda_{l} \mu_b$ b-epochs are stable, the first property of partition-resilience is ensured with high probability by Lemma \ref{lemma:totalnodes} and \ref{lemma:hon-peer-lower}, for a large enough $\lambda_n$. (For the first b-epoch, there exists at least $\lambda_p$ honest peers in each committee if the network is correctly bootstrapped.) Observe that if the join protocols of all the nodes that joined in the last $\lambda_{l} \mu_b$ are successful, then the requirements for the connections among peers are also ensured. Since the last $\lambda_{l} \mu_b$ b-epochs are bandwidth-adequate and the active directory is robust, then by Lemma \ref{lemma:join-succ}, all the (honest) nodes are successful with high probability, for a large enough $\lambda_j$.
\end{proof}

\begin{lemma} \label{lemma:netsizeest}
If b-epoch $e$ is stable and bandwidth-adequate, and the network is partition-resilient in b-epoch $e$, then whp, the (honest) peers can calculate the quantity $M'_e$ such that $M'_e \in [(1-\delta_{\mathit{err}})(1-\rho)M^L_e / \mu_b, (1+\delta_{\mathit{err}})\mu_b M^H_e]$.
\end{lemma}
\begin{proof}
Unlike other protocols, the network size estimation protocol is minimally affected by $\Delta$-synchrony. (Recall that there can be a difference of $\mu_b$ blocks in the chain lengths of honest peers.) This is because the node joins are bound to the block used in their proofs. Thus, the existing peers consider all the node joins in phase 1 even though they may start and end the phase 1 of b-epoch at possibly different times. Let $G'_e$ and $G_e$ be the estimate and actual number of nodes that joined in phase 1 of b-epoch $e$ respectively. First, let us show that whp, all the honest peers can calculate the quantity $G'_e$, and $G'_e \in [(1-\delta_1)G_e, (1+\delta_1)G_e]$ where $\delta_1 < 1$ is a small positive constant.

Let $J_e$ be the expected number of new nodes that can be generated in phase 1 of b-epoch $e$. Recall that the node difficulty threshold, $p_n = (\lambda_n \log N) / (q \alpha)$, is fixed. Since by design, phase 1 of a b-epoch is $\Theta(\alpha)$, for a large enough $\lambda_n$, using Chernoff bounds, $G_e \in [(1-\delta_2)(1-\rho)J_e, (1+\delta_2)J_e]$, for a small positive constant $\delta_2 < 1$. The lower bound has a factor of $(1-\rho)$ because the Byzantine peers may choose not to join any of their nodes (or even mine for nodes). For ease of exposition, let the aforementioned bounds be $G_e \in [C_1\lambda_{n}\mathcal{C}_e\log N, C_2\lambda_{n}\mathcal{C}_e\log N]$, where $C_1$ and $C_2$ are appropriate constants.

Note that $k_c \geq 1$ random committees broadcast the newly joined node IDs (Step 3 of Algorithm \ref{alg:sizeest} description in Section \ref{subsec:netsizeest}). All the peers get to know the number of new nodes that have joined each of those $k_c$ committees in at most $\Delta$ rounds because the network is partition-resilient and b-epoch $e$ is bandwidth-adequate. The peers need to estimate $G_e$ from that information. This is analogous to showing that if $m = G_e$ balls are (uniform) randomly thrown into $n = \mathcal{C}_e$ bins. If $k_c$ bins are (uniform) randomly picked to see the number of balls in them, then the estimate of $m$ is within a (fixed) multiplicative error of $\delta_1 \in (0, 1]$, whp, assuming $m/n \geq C_1\lambda_{n}\log N$.

Let $X_{i, j}$ be a Bernoulli random variable such that $X_{i, j}=1$ if the $i^\mathrm{th}$ ball is thrown into bin $j$, and $X_{i, j}=0$ otherwise. For a given ball $i$, $X_{i, 1}, X_{i, 2}, \dots, X_{i, n}$ are zero-one random variables, meaning that $\forall j, X_{i, j} \in \{0, 1\}$, and their sum being equal to 1, $\sum_j X_{i, j} = 1$. By Lemma 8 in \cite{dubhashi1996balls}, these random variables are negatively associated. Let $\mathbf{X}_i = \{ X_{i, j}\}^{j=n}_{j=1}$, then it is easy to see that $\mathbf{X}_i$ and $\mathbf{X}_j$, for $i \neq j$, are independent. Then, by Proposition 7 in \cite{dubhashi1996balls}, the full vector, $\{X_{i, j}\}^{i=m, j=n}_{i=1, j=1}$, is also negatively associated.

Since $k_c$ bins are (uniform) randomly picked, by symmetry, let them be the first $k_c$ bins. Let $H_e = \sum^{i=m, j=k}_{i=1, j=1} X_{i, j}$. Since for all $j$, $ \mathbb{E}[\sum^{i=m}_{i=1} X_{i, j}] = m/n$, by linearity of expectation, $\mathbb{E}[H_e] = k_cm/n$. Using both sides of Chernoff bounds,
\begin{align*} 
    \mathrm{Pr}[H_e \not\in [(1-\delta_1)mk_c/n, (1+\delta_1)mk_c/n]] &= Pr[H_e\mathcal{C}_e/k_c \not\in [(1-\delta_1)G_e, (1+\delta_1)G_e]]\\
    & \leq 2\exp(-mk_c\delta_1^2/3n)\\
    & \leq 2\exp(-C_1\lambda_{n}k_c\log N\delta_1^2/3).
\end{align*}
For a large enough $\lambda_{n}$, $C_1\lambda_{n}k_c\delta_1^2/3 > 1$. This proves that whp, all the honest peers can calculate the quantity $G'_e$, and $G'_e \in [(1-\delta_1)G_e, (1+\delta_1)G_e]$ where $\delta_1 < 1$ is a small positive constant.

Recall that the actual number of nodes that are generated in phase 1 of b-epoch $e$ is $G_e \in [(1-\delta_2)(1-\rho)J_e, (1+\delta_2)J_e]$. By blockchain liveness, the time elapsed in phase 1 of any b-epoch is at least $\alpha_1 / \mu^2_b$ rounds and at most $\alpha_1$ rounds. Hence, $p_n \alpha_1 M^L_e / \mu^2_b \leq J_e \leq p_n \alpha_1 M^H_e$, and this implies that $G_e \in [(1-\delta_2)(1-\rho)p_n \alpha_1 M^L_e / \mu^2_b, (1+\delta_2)p_n \alpha_1 M^H_e]$ where $M^L_e$ and $M^H_e$ are the minimum and maximum network sizes in phase 1 of b-epoch $e$. If the constants $\delta_1$ and $\delta_2$ are appropriately chosen, and it is given to us that $G'_e \in [(1-\delta_1)G_e, (1+\delta_1)G_e]$, this would imply that $G'_e \in [(1-\delta_{\mathit{err}})(1-\rho)p_n \alpha_1 M^L_e / \mu^2_b, (1+\delta_{\mathit{err}})p_n \alpha_1 M^H_e]$.

Since $p_n, \alpha_1$ and $\mu_b$ are known to all the peers, the (honest) peers calculate the quantity $M'_e = G'_e \mu_b /p_n \alpha_1$ so that $M'_e \in [(1-\delta_{\mathit{err}})(1-\rho)M^L_e / \mu_b, (1+\delta_{\mathit{err}})\mu_b M^H_e]$, thus proving the lemma.
\end{proof}

\begin{theorem} \label{the:stabilizesucc}
If b-epoch $e$ is stable and bandwidth-adequate, and the network is partition-resilient in b-epoch $e$, then whp, b-epoch $e$ is synchronized.
\end{theorem}
\begin{proof}
After phase 1 of b-epoch $e$, all the honest peers receive the set of nodes that joined a (randomly selected) committee (Step 3 of Algorithm \ref{alg:sizeest} description in Section \ref{subsec:netsizeest}) as the network is partition-resilient and b-epoch $e$ is bandwidth-adequate. In Step 3 of Algorithm \ref{alg:sizeest} description in Section \ref{subsec:netsizeest}, due to Lemma \ref{lemma:netsizeest}, whp, all the honest peers can calculate the quantity $M'_e$ such that $M'_e \in [(1-\delta_{\mathit{err}})(1-\rho)M^L_e / \mu_b, (1+\delta_{\mathit{err}})\mu_b M^H_e]$.

Finally, by the fairness assumption, with high probability, in a segment of length of $L = \lambda_{\mathit{lb}}\log N$, the fraction of blocks held by the honest peers is at least $(1-\delta')(1-\rho)$, where $\delta'$ is a small constant that depends on $\lambda_{\mathit{lb}}$. The parameters are set such that $(1-\delta')(1- \rho) > 0.5$, which means that there is an honest majority of blocks.
\end{proof}

\begin{lemma} \label{LEMmaintaingoodest}
If b-epoch $e$ is synchronized, then b-epoch $e$ has a good estimate ratio.
\end{lemma}

\begin{proof}
All (honest) peers adopt the network size estimate after phase 1 of b-epoch $e$ by setting $L'_e = M'_e$, where $M'_e \in [(1-\delta_{\mathit{err}})(1-\rho)M^L_e / \mu_b, (1+\delta_{\mathit{err}})\mu_b M^H_e]$, and $M^L_e$ and $M^H_e$ are the minimum and maximum network sizes in phase 1 of b-epoch $e$ and $\delta_{\mathit{err}} < 1$ is a small positive constant. Notice that $M^H_e/2 \leq L_e \leq 2M^L_e$ as the network size can change by a multiplicative factor of $2$ in any b-epoch. (This is because the number of rounds in any b-epoch is at most $\alpha$ rounds by the blockchain liveness.)

To maximize the estimate ratio of b-epoch $e$, we need to minimize $L'_e$ and maximize $L_e$,
\begin{align*}
    \frac{L_e}{L'_e} = \frac{2 M^L_e}{(1-\delta_{\mathit{err}})(1-\rho)M^L_e / \mu_b} = \frac{2 \mu_b}{(1-\rho)(1-\delta_{\mathit{err}})}.
\end{align*}
This upper bound is actually attained when the Byzantine peers do not join their nodes in phase 1 (thereby reducing the total estimate of the hash power by a factor $(1-\rho)$) and the total hash rate remains the same until phase 1, and increases by a factor of $2$ by the end of the b-epoch (which means that this increase in hash rate was not captured by the estimation algorithm in phase 1).

To minimize the estimate ratio of b-epoch $e$, we need to maximize $L'_e$ and minimize $L_e$,
\begin{align*}
    \frac{L_e}{L'_e} = \frac{M^H_e}{2(1+\delta_{\mathit{err}}) \mu_b M^H_e} = \frac{1}{2\mu_b (1+\delta_{\mathit{err}})}.
\end{align*}
And this lower bound is actually attained when all the Byzantine peers (mine and) join their nodes in phase 1 and the total hash rate remains the same until phase 1, and decreases by a factor of $2$ by the end of the b-epoch.
\end{proof}

\begin{theorem} \label{the:stablebepoch}
If the network is partition-resilient in b-epochs $e-1$ and $e-2$, and if b-epochs $e-1$ and $e-2$ are synchronized, and if the b-epochs $e-1$ is bandwidth-adequate, then b-epoch $e$ is a stable b-epoch, for $e \geq 3$. If the network is appropriately bootstrapped (Section \ref{sec:stablenetsize}), then b-epochs 1 and 2 are stable b-epochs.
\end{theorem}
\begin{proof}
Let us first consider a b-epoch $e \geq 3$. For b-epoch $e$ to be a stable b-epoch, the dimension of the hypercube should be appropriately set according to the network size in any round in b-epoch $e$. Since there is a ``lag'' of one b-epoch in increasing the dimension (see Section \ref{subsec:dimchange}), the decision to change the dimension of the hypercube must be taken by all the peers at the end of b-epoch $e-2$; this can be done because the b-epoch $e-2$ is synchronized. By the end of b-epoch $e-1$, the new number of committees is agreed upon by all peers again because the b-epoch $e-2$ is synchronized. As the b-epoch $e-1$ is bandwidth-adequate, and the network is partition-resilient in b-epoch $e-1$, the appropriate dimension of the hypercube in b-epoch $e$ is adopted such that the network size in any round $r$ of b-epoch $e$ satisfies the bounds given in Definition \ref{DEFstableb-epoch}.

Section \ref{sec:stablenetsize} provides details about bootstrapping the network such that Definition \ref{DEFstableb-epoch} is satisfied. Moreover, the dimension of the hypercube need not be changed in both b-epoch 1 and 2 because the network size can change by at most a factor of $2$ (as there are at most $\alpha$ rounds in a b-epoch via blockchain liveness). In other words, the network size will remain within the bounds given in Definition \ref{DEFstableb-epoch}. Thus, b-epoch 1 and 2 are stable.
\end{proof}

\begin{theorem}\label{theorem:main}
The overlay network is partition-resilient and each honest peer sends or receives $O(\log^3 N)$ messages per round, for a polynomial number of rounds with high probability.
\end{theorem}

\begin{proof}
The goal is to maintain partition-resilience over a polynomial number of b-epochs once the network is appropriately bootstrapped at time zero. We carefully exploit the dependencies of the theorems described so far to prove the statement. We describe a series of events that turn out to be useful.
\begin{itemize}
    \item $\mathtt{P}_i$ is the event that the network is partition-resilient in b-epoch $i$.
    \item $\mathtt{Q}_i$ is the event that b-epoch $i$ is synchronized.
    \item $\mathtt{R}_i$ is the event that b-epochs $i, i-1, \dots, z$ where $z = \mathrm{max}(1, i-\lceil \lambda_{l} \mu_b \rceil)$, are stable.
    \item $\mathtt{T}_i$ is the event that b-epochs $i, i-1, \dots, z$ where $z = \mathrm{max}(1, i-\lceil \lambda_{\mathit{dl}} \mu_b \rceil)$, are bandwidth-adequate.
    \item $\mathtt{U}_i$ is the event that the active directory is robust in b-epochs $i, i-1, \dots, z$ where $z = \mathrm{max}(1, i-\lceil \lambda_{l} \mu_b \rceil)$.
    \item $\mathtt{S}_i$ is the event that b-epoch $i$ is ``successful'', meaning that, $\mathtt{S}_i = \mathtt{P}_i \cap \mathtt{Q}_i \cap \mathtt{R}_i \cap \mathtt{T}_i \cap \mathtt{U}_i$.
\end{itemize}

Both b-epoch 1 and 2 do not require dimension change if the network is correctly bootstrapped (Section \ref{sec:stablenetsize}). As mentioned in Theorem \ref{the:stablebepoch}, b-epoch 1 is a stable b-epoch, i.e., $\mathtt{R}_{1}$ occurs. If $\mathtt{R}_{1}$ happens, then $\mathtt{T}_{1}$ occurs with high probability by Theorem \ref{th:totalbandwidth}. And if $\mathtt{T}_{1}$ happens, then $\mathtt{U}_i$ occurs with high probability due to Theorem \ref{the:actdir-robust}. If $\mathtt{R}_{1}$, $\mathtt{T}_{1}$ and $\mathtt{U}_i$ happen, then $\mathtt{P}_{1}$ happens with high probability because of Theorem \ref{the:partres}. This also means that $\mathtt{Q}_{1}$ happens with high probability due to Theorem \ref{the:stabilizesucc}. Therefore, applying a union bound, $\mathtt{S}_{1}$ happens with high probability. The same chain of arguments also holds for b-epoch $2$.

For a b-epoch $i$ where $i \geq 2$, if there is a dimension change, the system, intuitively, gets bootstrapped again (during the transformation b-epoch). The key observation in Lemma \ref{lemma:hon-peer-lower} and Theorem \ref{the:partres} is that the lower bound on the number of honest peers for partition resilience relies only on the honest node joins of previous b-epoch. This is important because even though after a dimension change (adopting the new hypercube), the nodes that joined before the transformation b-epoch are considered invalid, Lemma \ref{lemma:hon-peer-lower} and Theorem \ref{the:partres} still apply to ensure partition-resilience for the nodes in the new hypercube. 

Given that the events $\mathtt{S}_{i}$ and $\mathtt{S}_{i-1}$ have occurred, by Theorem \ref{the:stablebepoch}, $\mathtt{P}_{i} \cap \mathtt{Q}_{i} \cap \mathtt{R}_{i}$ and $\mathtt{P}_{i-1} \cap \mathtt{Q}_{i-1} \cap \mathtt{R}_{i-1}$ imply that b-epoch $i+1$ is a stable b-epoch, in other words, event $\mathtt{R}_{i+1}$ happens. Once $\mathtt{R}_{i+1}$ has taken place, then again, by a similar chain of arguments, the event $\mathtt{S}_{i+1}$ happens with high probability. Applying a union bound over a polynomial number of b-epochs, the network is partition-resilient and each honest peer sends or receives $O(\log^3 N)$ messages, with a high probability.
\end{proof}

\subsection{Recovery analysis} \label{subsec:rec-ana}

\begin{definition}
An honest peer is said to be \text{\normalfont corrupted} if it either stops functioning (i.e., leaves the network) or becomes Byzantine in which case it is controlled by the adversary.
\end{definition}

\begin{definition}
A bucket is said to have \text{\normalfont failed} if more than $1/2$ fraction of honest peers in it are corrupted.
\end{definition}

\begin{definition}\label{def:safe-comm}
Let the set of buckets in the active directory that are responsible for a committee $C$ be denoted as $B_C$. A committee $C$ is said to be \text{\normalfont safe} if:
\begin{enumerate}
    \item It has at least $20 \log N$ honest peers.
    \item No bucket in $B_C$ has failed.
\end{enumerate}
\end{definition}

\begin{definition}
Two committees $C_1$ and $C_2$ are said to be \text{\normalfont connected} if:
\begin{enumerate}
    \item Each honest peer in $C_1$ is connected to $\Omega(\log N)$ honest peers in $C_2$.
    \item Each honest peer in $C_2$ is connected to $\Omega(\log N)$ honest peers in $C_1$.
\end{enumerate}
\end{definition}

\begin{definition}\label{def:cat-fail}
Let $S$ denote a set of safe committees. Let $G_S$ be the graph where  vertices correspond to safe committees in $S$ and edges represent the connection between two safe committees. The overlay network is said to experience an \text{\normalfont $(\epsilon, \delta)$-catastrophic failure} in b-epoch $e$ if the following events occur.
\begin{enumerate}
    \item There are at least $(1-\epsilon) \mathcal{C}_{e}$ safe committees for a small constant $\epsilon > 0$.
    \item At most $\delta$ fraction of peers get corrupted for a small constant $\delta > 0$.
    \item There exists a graph $G_S$ with a diameter at most $2 \log N$, in which at least $\mu_n a $ fraction of honest peers are present\footnote{For example, a blockchain protocol may progress with 70\% of total honest peers but no Byzantine peers, or with 90\% of total honest peers and 20\% of total peers being Byzantine. $\mu_n$, which depends on the maximum number of Byzantine peers, captures the fact that the blockchain provides its guarantees even if not all the honest peers participate. Note this is a minimal requirement to recover from a (massive) eclipse attack.}, where $|S| \geq a b \mu_n \mathcal{C}_{e}$ for $a > 1/(1-\epsilon')$ and $b = 1/(1-\epsilon'')$ for an arbitrarily small positive constant $\epsilon'$ and $\epsilon'' \leq O((\log^{1/2} N) / N^{1/2y})$.
\end{enumerate}
\end{definition}

\begin{lemma}\label{lemma:multi-hyp} (Restated, Theorem 1 of \cite{datar2002butterflies}) 
No matter which $f$ nodes are made faulty in a multi-hypercube with $n$ nodes and $(\alpha, \beta)$ expansion, there are at least $n - \frac{\beta f}{\beta - 1}$ nodes that  have $\log n$-length path to at least $n - \frac{f}{\alpha(\beta - 1)}$, such that all the nodes in the path are not faulty.
\end{lemma}


\begin{lemma}\label{lemma:const-exp-2}
(Restated, Lemma 4.6 of \cite{fiat2007censorship}) Let $l, l', r, 'r, d, k, \lambda$ and $n$ be any positive values where $l' \leq l$, $r' \leq r$, $0 < \lambda < 1$ and
\begin{equation*}
    d \geq \frac{2 r}{r' l' (1-\lambda)^2}\left( l'\ln\left(\frac{le}{l'}\right) + r'\ln\left(\frac{re}{r'}\right) + k\ln n \right)
\end{equation*}
Let $G$ be a random bipartite multigraph with left side $L$ and right side $R$ where $|L| = l$, $|R| = r$ and each node in $L$ has $d$ random neighbours in $R$. Then, with probability at least $1-n^{-k}$, for any subset of $L' \subset L$ where $|L'| = l'$, there is no set $R' \subset R$ where $|R'| \subset R$, such that all nodes in $R'$ share less than $\lambda l' d / r$ edges with $L'$.
\end{lemma}

\begin{lemma}\label{lemma:low-dia}
No matter which $f$ nodes are made faulty in a multi-hypercube with $n$ nodes and $(\alpha, \beta)$ expansion, there exists a connected component of at least $n - 3f/2$ non-faulty nodes with a diameter of at most $2 \log n$, for $\alpha(\beta - 1) \geq 2/3$ and $\beta \geq 3$.
\end{lemma}
\begin{proof}
This follows directly from Lemma \ref{lemma:multi-hyp}. Viewing the multi-hypercube as a $n$-node multi-butterfly network where a row of switches is simulated by a node, the theorem says that any of the $\left(n - \frac{\beta f}{\beta - 1}\right)$ (non-faulty) inputs can reach any of the (non-faulty) $\left(n - \frac{f}{\alpha(\beta - 1)}\right)$. For $\alpha(\beta - 1) \geq 2/3$ and $\beta \geq 3$, consider a connected component $c$ of any one such (non-faulty) input and the $n - 3f/2$ (non-faulty) outputs including the intermediate switches; each pair of them can reach other by a path of length at most $2 \log n$. In case of a multi-hypercube, since a row (consisting of one input, one output and $\log n$ switches) is simulated by one node, there are at least $n - 3f/2$ \emph{distinct} nodes (because one output is simulated by one node and there are $n - 3f/2$ such outputs) in $c$ that have a diameter of at most $2 \log n$.
\end{proof}

\begin{lemma}\label{lemma:advers-fail-good-comm}
For $e \geq 2$ and constant $\lambda_p > 0$, if at most $\delta$ fraction of peers are corrupted in a stable b-epoch $e-1$, then in b-epoch $e$, there are at least $(1-\epsilon')\mathcal{C}_e$ committees that are assigned with at least $\lambda_p \log N$ honest peers with high probability, where $\epsilon' \leq O(1 / \log N)$ and $\delta < 1$ is a small positive constant.
\end{lemma}
\begin{proof}
First, we get a lower bound on the number of honest peers that stay in the entire b-epoch $e-1$. By stable b-epoch property and the fact that at most $\rho$ fraction are Byzantine, this is at least $(1-\rho) \mathcal{C}_{e-1} / \lambda_s$. Then, we find a lower bound on the number of new nodes generated in b-epoch $e-1$ by each of these peers. Recall that the node difficulty threshold, $p_n = (\lambda_n \log N) / (q \alpha)$, is fixed. By blockchain liveness (for number of rounds in a b-epoch), the expected number of new honest nodes, in b-epoch $e-1$, is at least $(\lambda_n \log N) / \mu^2_b$. Applying Chernoff bounds and union bound, the total number of honest nodes by each of those peers generated is at least $d = \Theta(\lambda_n \log N)$ with high probability. (We drop the other constants unless necessary, as $\lambda_n$ controls the failure probability.)

These nodes get randomly mapped to $\mathcal{C}_{e}$ committees. (Note that $\mathcal{C}_e = \mathcal{C}_{e-1}$ if b-epoch $e-1$ is not a transformation b-epoch, and even otherwise, $\mathcal{C}_e = \Theta(\mathcal{C}_{e-1})$.) Here, in Lemma \ref{lemma:const-exp-2}, we can consider these peers as the left side $L$ and the committees in b-epoch $e$ as the right side $R$ where each vertex in $L$ has at least $d$ random neighbours in $R$. In other words, with probability less than $n^{-k}$ for some constant $k > 2$, there exists a set $L' \subset L$ such that $l' = \delta l$ and a set $R' \subset R$ of size $(1-\epsilon') \mathcal{C}_e$ where each vertex has less than $\lambda_p c_m \log N$ nodes, for a positive constant $c_m$. Specifically, $\lambda_n$ is set large enough so that $\lambda l' d / r \geq \lambda_p c_m \log N$. As every other parameter is fixed, obeying the condition of Lemma \ref{lemma:const-exp-2} by (again) setting $\lambda_n$ large enough, with high probability, we get that there are at least $(1-\epsilon')\mathcal{C}_e$ committees that are assigned with at least $\lambda_p c_m \log N$ honest nodes even after $\delta$ fraction of peers get corrupted.

Finally, by Lemma \ref{lemma:peer-node-comm}, each peer controls at most $c_m$ nodes in a committee with a high probability. Thus, in b-epoch $e$, there are at least $(1-\epsilon')\mathcal{C}_e$ committees that are assigned with at least $\lambda_p \log N$ honest peers with high probability.
\end{proof}

\begin{lemma}\label{lemma:replenish}
In any round $r$, if there are at least $\mu_n a I$ honest peers in a set $S$ of committees, the following statements hold with high probability, where $|S|, |S'| \geq \mu_n a b C$, $C$ is the number of committees, $\mu_n a b < 1$, $a > 1/(1-\epsilon')$ and $b = 1/(1-\epsilon'')$ for $\epsilon'' \leq O((\log^{1/2} N) / N^{1/2y})$ and an arbitrarily small positive constant $\epsilon'$. Here, $I$ is the total number of honest peers at round $r$. $J_{r'}$ and $L_{r'}$ are the number of new honest peers that joined the network and the number of peers that left the network respectively, from round $r$ to $r'$.
\begin{enumerate}
    \item There are at least $\mu_n a (I + J_{r_1} - L_{r_1})$ honest peers at round $r + r_1$ for any $\alpha / \mu^2_b \leq r_1 \leq \alpha$ in set $S'$ of committees.
    \item For any round $r \leq r_2 \leq \alpha$, there are at least $\mu_n  (I + J_{r_2} - L_{r_2})$ honest peers in the set $S$ of committees.
\end{enumerate}
\end{lemma}
\begin{proof}
The Lemma says that there exists a large enough set of committees that can be replenished with a set of distinct honest peers in a span of a b-epoch. And that in any round in the b-epoch, despite joins and leaves, a certain fraction of honest peers are maintained in those set of committees. Note that this Lemma does not deal with successful joins, robust active directory, etc., but only talks about peers getting assigned to committees. (They will be dealt with in Theorem \ref{theorem:rec-time}.)

Let us prove the first property. Let $L^1_{r_1}$ be the number of honest peers in round $r$ that left the network by round $r_1$. Let $L^2_{r_1}$ be the number of honest peers from the set of peers that joined in subsequent rounds left the network by round $r_1$. Thus, $L_{r_1} = L^1_{r_1} + L^2_{r_1}$. Let $J^1_{r_1}$ be the number of (new) honest peers that joined and stayed till round $r_1$. Let $J^2_{r_1}$ be the number of (new) honest peers that joined and left the network by $r_1$. Thus, $J_{r_1} = J^1_{r_1} + J^2_{r_1}$.

Note that $(I + J_{r_1} - L_{r_1}) = (I + J^1_{r_1} - L^1_{r_1})$. For a large enough $\lambda_n$, each of the $(I - L^1_{r_1})$ get at least one new (valid) node with high probability (via Chernoff bounds). Since $J^1_{r_1}$ peers joined the network, they must have joined with at least one (valid) node. Thus, out of the $(I + J^1_{r_1} - L^1_{r_1})$ peer joins, on expectation, $\mu_n a b (I + J^1_{r_1} - L^1_{r_1})$ go to the set $S'$ of committees (as each node is randomly mapped to a committee). Thus, applying Chernoff bounds, for $\epsilon'' \leq O((\log^{1/2} N) / N^{1/2y})$, there are at least $\mu_n a (I + J_{r_1} - L_{r_1})$ honest peers at round $r + r_1$ with high probability.

Consider a large enough $s = \Theta(\log N)$. We batch together $s$ leaves and $s$ joins over the rounds until round $r_2$ to prove the second property. Batching the joins together is not difficult as they are independent of each other. Batching the leaves together requires a little care. Recall that the sequence of honest leaves are obliviously specified by the churn adversary. Let us say that there are $n$ peers at round $r$. And, for some round $r' \geq r$, there are $l$ honest peer leaves and $j$ honest peer joins from round $r$ until round $r'$. Then, focusing on $r'$, we can claim that any set of $l$ honest peers (out of $n+j$ peers) has equal probability of having left.

Applying Chernoff bounds for a batch of $s$ honest leaves, there are at most $(1+\epsilon')\mu_n a s$ honest leaves from the set $S$ of committees with high probability for a small $\epsilon' > 0$ and large enough $s$. Similarly, applying Chernoff bounds for a batch of $s$ honest joins, there are at least $\mu_n a b (1-\epsilon')s \geq \mu_n a (1-\epsilon')s$ honest joins from the set $S$ of committees with high probability for a small $\epsilon' > 0$ and large enough $s$.

Let the number of honest peers at round $r$ in the set $S$ of committees be denoted as $H_{r}$. Then, for any round $r \leq r_2 \leq \alpha$,
\begin{align*}
   H_{r_2} &\geq \mu_naI - \mu_n a(1+\epsilon')L_{r_2} - O(\log N) + \mu_n a(1-\epsilon')J_{r_2} - O(\log N)\\
   &= \mu_n  (I + (a-1)I - L_{r_2} - (a(1+\epsilon')-1)L_{r_2} + J_{r_2} + (a(1-\epsilon') - 1)J_{r_2}) - O(\log N)\\
   &= \mu_n (I - L_{r_2} + J_{r_2}) + \mu_n ((a-1)I - (x(1+\epsilon')-1)L_{r_2} + (a(1-\epsilon') - 1)J_{r_2}) - O(\log N).
\end{align*}
For $a \geq 1/(1-\epsilon')$, we can ignore the second join term as that would only increase the RHS, 
\begin{equation*}
   H_{r_2} \geq  \mu_n (I - L_{r_2} + J_{r_2}) + \mu_n ((a-1)I - (a(1+\epsilon')-1)L_{r_2}) - O(\log N).
\end{equation*}
By using the fact that if the number of honest peers at round $r$ is $I$, then at any $r \leq  r_2 \leq \alpha$, then at most half the number of those peers can leave the network, i.e., considering all peers are honest (which maximizes $L_{r_2}$), we get that $L_{r_2} \leq I/2$. Thus, for $a > 1/(1-\epsilon')$,
\begin{align*}
   H_{r_2} &\geq  \mu_n (I - L_{r_2} + J_{r_2}) + \mu_n I((a-1) - (a(1+\epsilon')-1)/2) - O(\log N)\\
   &\geq \mu_n (I - L_{r_2} + J_{r_2}) + \mu_nI\left( a\left(\frac{1-\epsilon'}{2}\right) - \frac{1}{2} \right) - O(\log N)\\
   &\geq  \mu_n (I - L_{r_2} + J_{r_2}).
\end{align*}

\end{proof}

\begin{theorem}\label{theorem:rec-time}
If the overlay network experiences an $(\epsilon, \delta)$-catastrophic failure, then it becomes partition-resilient in a constant number of epochs with high probability.
\end{theorem}

\begin{proof}
The high-level intuition for attaining recovery is to replace the entire active directory by a new one (that has no bucket failures). If that is posssible, then new active directory facilitates honest nodes joins to committees as required. Catastrophic failure is defined in a way that there is one large connected component of sufficient number of honest peers with a low diameter, that is responsible for the progress of the blockchain. First, we need to prove that the component gets replenished with new (honest) peers as peers join and leave. Then, we focus on that component of committees and show that the lemmas and theorems in Section \ref{sec:full-ana} hold for them, albeit for some small changes. If we can show that the overlay can rely on those subset of committees until a new active directory is formed, then after one additional b-epoch, the overlay reverts to being (fully) partition-resilient as in Section \ref{sec:full-ana}. (Recall that an active directory consists of $O(\alpha / \beta)$ blocks, amounting up to a $O(\alpha)$ rounds, via blockchain liveness.) The main technical difficulty lies in handling dimension change and showing that the new topology has a large connected component of committees with a low diameter, having a sufficient number of honest peers. And to handle that difficulty, we rely on the properties of multi-hypercube \cite{datar2002butterflies} and bipartite expanders \cite{fiat2002censorship, fiat2007censorship}.

Let us say that it takes at most $K$ half-lives for an active directory to be formed. We split the proof into three cases. (Case 2 and 3 can be merged together, but we keep them separate for ease of understanding.)

\textit{Case 1.} In this case, the overlay does not initiate a dimension change in any b-epoch $e, e+1, \dots, e+K+1$. If $\mu_n a b < (1-\epsilon)$ in Lemma \ref{lemma:replenish} where $a, b$ as in Definition \ref{def:cat-fail}, the set $S$ of committees have a sufficient number of honest peers in b-epoch $e$ and the subsequent b-epochs with high probability. We now focus on the set $S$ of committees (and their buckets) and show that the lemmas and theorems in Section \ref{sec:full-ana} also hold after a catastrophic failure. The lemmas related to bounds on the number of node assignments to committees, number of nodes controlled by a peer and the number of join requests handled by a peer (and the total communication cost) apply here without any change. Instead of showing that the entire active directory is robust, Theorem \ref{the:actdir-robust} can be used to show that the non-failed buckets have the required properties of robustness (Definition \ref{def:actrobust}). This can be done by increasing $\lambda_{d}$ by a factor of 2, after which the same analysis holds in Theorem \ref{the:actdir-robust}, as at most 1/2 fraction of honest peers in a non-failed bucket get corrupted. Thus, the node joins for those buckets are successful.

The network is not partition-resilient, even for the set $S$ of committees, because by definition, safe committees have at least $20 \log N$ honest peers (and not at least $\lambda_{p} \log N$ honest peers). As the safe committees have at least $20 \log N$ honest peers, by Chernoff and union bounds, at least $\Omega(\log N)$ in each safe committee would stay for another 2 b-epochs with high probability. Thus, connectivity among committees in $S$ maintained during those two b-epochs. Then, by Lemma \ref{lemma:hon-peer-lower}, using the node joins in b-epoch $e+1$, we can show that the properties of partition resilience (Definition \ref{DEFpartres}) hold for the set $S$ of committees from b-epoch $e+2$ onwards. (The definition of partition-resilience is deliberately made to be topology-oblivious, as in, it can also apply to any subset of committees.)

The tricky part is to show that partition-resilience properties can be maintained for $S$ if network size is allowed to significantly vary over time. Specifically, a small modification is required for Algorithm \ref{alg:sizeest} used to estimate the network size. A random committee is selected using the hash of the block that determines the end of phase 1 of that b-epoch. Instead of choosing a random committee, we make a small modification to the protocol to choose $\Theta(\log N)$ random committees. (For example, this can be done by appending $1, 2, \dots, \log N$ to the input, consisting of hash of the block, for the hash function, and obtaining the output for each of those inputs. In other words, the hash function can be used to generate the required verifiable randomness, where the computational cost is negligible. Recall that each peer can query the hash function $q > 0$ number of times in a round, where $q$ is substantially large.) Once a node receives the entry information of the nodes that belonged to those committees, it chooses the committee that encountered the maximum number of node joins. If this modification is done, then with a high probability, that chosen committee would belong to $S$ because $|S| > \mathcal{C}_e$. (The nodes consider the committee with maximum number of node joins because the failed committees can only under-represent node joins.) Thus, the rest of the lemmas and theorems regarding network size estimation and b-epoch synchronization hold if the network size estimation can be done securely. (Also, network size estimation is required for the subsequent cases in which dimension change needs to be done.)

\textit{Case 2.} In this case, b-epoch $i$ is a transformation b-epoch for $i > e$. The proof arguments from Case 1 also carry over to this case until b-epoch $i$. The tricky part is to handle dimension change, as it is important to maintain a set $S'$ in the new topology that has similiar properties as $S$ so that the proof arguments from Case 1 can also apply from b-epoch $i+1$ onwards. Our first observation is that at most $\epsilon$ fraction of buckets have failed in a directory. This is because if there are more bucket failures, then this would result in a number of committee failures greater than $\epsilon \mathcal{C}_e$. Since this applies to all the directories in the active directory, at most $\epsilon$ buckets fail in the active directory during a catastrophic failure. This means that at most $\epsilon$ committees have failed in the next topology as well, as the committee-directory mappings are such that each bucket is responsible for the same number of committees, and the sets of committees that any two buckets are responsible for, are disjoint. We now rely on the guarantees provided by the multi-hypercube. By Lemma \ref{lemma:multi-hyp}, at least $\mathcal{C}_{i+1}(1-(3 \epsilon/2))$ form a connected component\footnote{Refer to Definition \ref{def:safe-comm} for connectivity between committees.} with a diameter of at most $2\log N$. If $(1-(3 \epsilon/2)) > \mu_n a b$ as in Definition \ref{def:cat-fail}, then using Lemma \ref{lemma:replenish}, the same arguments carry over to the next topology. This applies to multiple dimension changes that can occur over the $K+1$ b-epochs.

\textit{Case 3.} In this case, b-epoch $e$ itself is a transformation b-epoch. We cannot just rely on the multi-hypercube as the honest peers that replenish the committees in $S'$ in the next topology, may get corrupted during (or towards the end of) b-epoch $e$. We first compute an upper bound on the number of committees that can have less than $\lambda_p \log N$ honest peers if in total $\delta$ fraction of honest peers are corrupted. By Lemma \ref{lemma:advers-fail-good-comm}, there are at most $\delta' \leq O(1/\log N)$ fraction of committees that have less than $\lambda_p \log N$ honest peers. Thus, if $(1-(3 \epsilon/2) - \delta') \geq \mu_n a b$ as in Definition \ref{def:cat-fail}, then using Lemma \ref{lemma:replenish}, there are a sufficient number of honest peers in $S'$ throughout b-epoch $i+1$ by Lemma \ref{lemma:replenish} with high probability. In other words, the overlay can rely on the set $S'$ for the progress of the blockchain for b-epoch $i+1$. Moreover, as in previous case, there are at most $\epsilon$ fraction of committees fail due to bucket failures. From b-epoch $i+2$ onwards, the committees that had less than $\lambda_p \log N$ honest peers in b-epoch $i+1$, but whose buckets had not failed, have at least $\lambda_p \log N$ honest peers with high probability by Lemma \ref{lemma:hon-peer-lower}. Thus, the set $S'$ again resorts to the large connected component of committees guaranteed by the multi-hypercube. Then, from b-epoch $i+2$ onwards, the overlay can rely on that set of committees for the progress of the blockchain, where the same proof arguments of Case 1 apply.

Thus, in all these cases, the overlay network becomes partition-resilient in a constant number of b-epochs with high probability.

\end{proof}

{
\bibliographystyle{prahladhurl}
\bibliography{main}
}

\appendix

\section{More on Recovery}\label{sec:more-rec}

In this section, we argue that the existing solutions for join-leave attacks cannot easily recover from the committee failures considered in Section \ref{sec:recovery}. Their join (and leave) protocols depend on the honest majority of committees\footnote{They are also referred to as quorums \cite{awerbuch2004group, awerbuch2009towards} or swarms \cite{fiat2005making} or clusters \cite{guerraoui2013highly} in the literature.}. Once a committee loses the honest majority, it can initiate new (malicious) joins, either to replenish itself with more malicious peers or make other committees fail. In particular, even when a small number of committees, say $O(\log N)$ committees have a malicious majority, then the network can continue to have at least $\Omega(\log N)$ committees with malicious majority over time. Before delving into the technical details, we strive to provide a high-level intuition for vulnerability to arbitrary committee failures in the existing solutions.

First, we specify the combination of network and join protocols for which recovery is hard to achieve. We are interested in a virtual network (that has low diameter, low degree and good expansion) of committees of $\Theta(\log N)$ peers. Typically, this is termed as a ``quorum topology'' \cite{young2010practical}. More specifically, the committees are functional units of the (virtual) network. We recall a few important invariants of a quorum topology \cite{young2010practical}.

\begin{definition} \label{def:quo-topo}
An overlay network is said to have a \text{\normalfont quorum topology} if it satisfies the following invariants.
\begin{enumerate}
    \item \textbf{Network of committees.} The overlay network is defined by a virtual graph of committees $G_C$ where vertices are committees and edges between nodes represent connections between committees.
    \item \textbf{Committee size.} Each committee consists of $\Theta(s)$ peers where $s = \Omega(\log N)$.
    \item \textbf{Membership.} Every peer belongs to at least one committee.
    \item \textbf{Intra-committee communication.} Every peer can communicate with all other members of its committees.
    \item \textbf{Inter-committee communication.} If $C_i$ and $C_j$ share an edge in $G_C$, then a peer belonging to $C_i$ can communicate directly with any member of $C_j$ and vice-versa.
\end{enumerate}
\end{definition}

In addition to the five invariants mentioned in Definition \ref{def:quo-topo}, the challenge is to maintain honest majority in each committee despite churn and a constant fraction of peers being Byzantine. The committees themselves form connections amongst themselves according to rules of an efficient network such as Chord \cite{awerbuch2004group, fiat2005making}, de Bruijn graph \cite{awerbuch2009towards}, an expander graph \cite{guerraoui2013highly} etc. Here, two committees are said to be \emph{connected} if each peer of a committee is connected to every other peer in the other committee (and vice versa). Thus, each committee is connected with a small number of other committees; we say that those other committees are its ``neighbouring'' committees. 

Although it is difficult to generalize the join protocols in the literature, we identify a few important steps or invariants that are central to all the existing join protocols (including this work) without getting to the details of how they are achieved.

\begin{definition} \label{def:churn-res-join}
The \texttt{JOIN} protocol for a network with a quorum topology is said to be \text{\normalfont churn-resilient} if it consists of the following steps (in the same order).
\begin{enumerate}
    \item \textbf{Random ID.} Firstly, when a peer $p$ joins the network, it is assigned a random committee $C_r$.
    
    \item \textbf{Placement.} Peers in $C_r$ and its neighbouring committees are informed about the new peer $p$ after which they verify that $p$ actually belongs to $C_r$.
    
    \item \textbf{Perturbation.} After the new peer gets placed in a random committee, the network is ``perturbed'' where at most $O(\polylog N)$ peers shift to different (typically random) committees.
\end{enumerate}
\end{definition}

To get an idea of how each of those steps can be achieved, we delve deeper into the existing solutions.
\begin{enumerate}
    \item \textbf{Random ID.} Typically, a new peer $p$ is assumed to know the contact address of an existing (honest) peer $q$. Peer $q$ belonging to committee $C_q$ informs all the peers in $C_q$ about the new peer $p$. For structured routable topologies \cite{awerbuch2004group, fiat2005making, awerbuch2009towards}, the peers within the committee $C_q$ run a distributed random number generation (RNG) protocol to determine a \textit{random location} in the network. (We use the term ``location'' because the previous solutions are based on placement of peers in the virtual (continuous) interval $[0, 1)$.) For the expander topology \cite{guerraoui2013highly}, the authors heavily rely on random walks over committees. A committee internally runs a distributed RNG protocol to determine the next walk step (random neighbour) until the walk ends. In our work and in \cite{jaiyeola2018tiny}, a peer gets bound to a random committee via hash function (which can be verified by any other peer).
    
    \textbf{Placement.} In structured routable topologies \cite{awerbuch2004group, fiat2005making, awerbuch2009towards}, the peers in $C_r$ are informed about $p$ by $C_q$ through an efficient route. Then, peers in $C_r$ send their contact information to peer $p$, and also inform their neighbouring committees about the membership of peer $p$. In the expander topology \cite{guerraoui2013highly}, the path used in the random walk is used to do the same. In our work, the bootstrapping service directly helps $p$ in contacting the $C_r$ and its neighbouring committees.
    
    \textbf{Perturbation.} Shuffling peers between committees is necessary upon join and leaves to maintain honest majority in all the committees against adaptive join-leave attacks \cite{awerbuch2004group, guerraoui2013highly}. Typically, this pertubation of network should be ``small''. For e.g., every peer can be assigned a random committee whenever a new peer joins the network, to maintain honest majority in all committees. But that would not be efficient at all. There are different ways to achieve small perturbation: adopting limited lifetime for peers \cite{awerbuch2004group}, k-rotation \cite{scheideler2005spread, fiat2005making}, cuckoo rule \cite{awerbuch2009towards}, and exchange all peers in $C_r$ with other random peers in the network \cite{guerraoui2013highly}. In our work, we rely on the limited lifetime method. The limited lifetime method is indeed a small perturbation \textit{per join} because the lifetime depends on churn rate (half-life) of the system. In other words, there is a linear number of joins and leaves in a half-life period, and therefore, setting a lifetime of a constant number of half-lives, essentially perturbs the network (forces re-joins) over a ``batch'' of joins and leaves.
\end{enumerate}

The fundamental problem in the existing solutions is that steps in Definition \ref{def:churn-res-join} heavily depends on the committees having honest majority. For example, the distributed RNG protocol is secure only if the committee has an honest majority. Moreover, the (random) placement of a new peer is verified through ``majority vote'' routing, i.e., when a committee receives a join request of a new peer from a majority of peers in one of its neighbouring committees, it verifies the randomness of the placement assuming such a majority vote occurred at each committee in the route, until the committee that executed the distributed RNG protocol. The majority vote verification is required for random walks solution too. Thus, honest majority of committees becomes crucial in proving the security of such join protocols. To this end, we strive to define ``committee-based'' joins that heavily rely on honest majority of committees.

\begin{definition}
A churn-resilient \texttt{JOIN} protocol for a network with quorum topology is said to be \text{\normalfont committee-based} if it has the following properties.
\begin{enumerate}
    \item A peer $p$ joins the network by contacting another peer $q$ which in turn informs its committee $C_q$ (introducing committee). $C_q$ initiates the random ID generation protocol so that $p$ gets placed in a random committee $C_r$ (placed committee).
    \item Honest majority in committees is required for carrying out and verifying the steps of the \texttt{JOIN} protocol (as mentioned in Definition \ref{def:churn-res-join}).
    \item The introducing and placed committees ($C_q$ and $C_r$) are responsible for initiating and carrying out network perturbation if it is not done via limited lifetime method.
\end{enumerate}
\end{definition}

Our work primarily differs in random ID and placement of the joining peer. The insight is that blockchain provides a globally known, network-generated and unpredictable \textit{input} to the hash function. If the input is globally known, then the hash function output can be verified by any peer. If the input is unpredictable, then the adversary cannot launch a pre-computation attack to populate a committee. If the input is generated by a set of committees (and not by the entire network), then that input can be corrupt if those set of committees have a malicious majority. Such an input is hard to generate using completely localized algorithms. The hash function is used to map the new peer to a random committee. The new peer is then placed in that committee by a secure bootstrapping service, constructed using the blockchain. Finally, by making use of the limited lifetime method to perturb the network, we do not rely on honest majority of committees for the security of join protocol.

As it is difficult to devise a single attack that works for the family of network and join protocols, we provide attacks to the existing solutions (primarily exploiting the reliance on honest majority). We want to show that a small number, say, $O(\log N)$ committees that have malicious majority to continually maintain malicious majority with low cost. Let $C_M$ be the committee having malicious majority that needs to get replenished with more malicious peers.
\begin{enumerate}
    \item \textbf{Structured routable topologies \cite{awerbuch2004group, fiat2005making, awerbuch2009towards}.} The key idea is that the $O(\log N)$ malicious majority committees simply route new malicious peers to $C_M$ with a join request. This attack is sufficient to keep adding new malicious peers into $C_M$ in \cite{awerbuch2004group}. In the k-rotation paper \cite{scheideler2005spread, fiat2005making}, the introducing committee selectively picks three locations to add and shift peers. In the cuckoo rule paper \cite{awerbuch2009towards}, the other $O(\log N)$ malicious majority committees pick disjoint constant-length arcs in $C_M$ for placing new malicious peers so that the perturbation does not affect the newly added peers.
    \item \textbf{Expander network topology \cite{guerraoui2013highly}.} Malicious committees can have edges amongst themselves to create join requests via random walks. Committee $C_M$ informs its neighbours about a new (malicious) peer. Instead of actually exchanging all other peers in $C_M$, malicious peers in $C_M$ can reserve a small fraction of the peers for exchanges between legitimate committees and replaces the rest of the peers with new malicious peers, thereby maintaining a malicious majority. More importantly, as new malicious peers get added to $C_M$, the committee can get split into two committees, each with malicious majority. (Locally splitting/merging is done to expand/shrink the network as the network size can vary polynomially over time.)
\end{enumerate}

For some specific attacks, statistical measures such as join rate, message rate, etc., can be used in practice, to detect and mitigate the effects of committee failures. However, not only are those measures not consistent over time in open P2P networks but even effectively carrying out such system-wide measurements is a hard problem. If committees are given the power of rejecting new peers, protesting against other committees, etc., then malicious majority committees can easily misuse such powers. Thus, recovery from arbitrary committee failures in networks that heavily rely on honest majority in committees, is difficult to achieve.

\section{Pseudocodes for Sub-routines} \label{appendix:subroutines}

\begin{algorithm}
\caption{\texttt{VERIFY\_PROOFS} protocol}
\label{alg:verifyproofs}
\begin{algorithmic}
\REQUIRE Let $R$ be the set of messages of the form $(m, e)$ received in this round, where $m$ is either \texttt{JOINING} and \texttt{REQ\_INFO} and $e$ is the entry information of the node that sent the message. Let $\mathbf{H}(.)$ denote the hash function. Let $\mathbf{Blk}(i)$ be a function that returns the block from the confirmed chain with block number $i$ if it exists, otherise, returns $\varnothing$. Let $l$ be the block number of the most recent block in the confirmed chain.
\ENSURE Return the subset of $R$ in which each message has a valid proof.
\STATE $V \leftarrow \{ \}$.
\FOR{each $(m, e) \in R$}
    \STATE $\mathit{blk\_num}, N_c, \mathit{net\_addr}  \leftarrow e$. \COMMENT{Retrieve entry information.}
    \STATE $\mathit{blk} \leftarrow \mathbf{Blk}(\mathit{blk\_num})$.
    \STATE $P_{\mathit{join}} \leftarrow \mathbf{H}(\mathbf{H}(\mathit{blk})  \mathbin\Vert \mathit{net\_addr} \mathbin\Vert N_c)$.
    \STATE $\mathit{c} \leftarrow$ Leftmost $\lceil \log \mathcal{C}_e \rceil$ bits of $P_{\mathit{join}}$.
    \STATE $\mathit{valid\_blk} \leftarrow (l - \mathit{blk\_num}) \leq \mu_s$ \AND ($\mathit{blk}$ is not $\varnothing$).
    \IF{the verification is for the directory}
    \STATE $C_n \leftarrow$ Set of IDs of neighbouring committees of committee $\mathit{c}$.
    \STATE $C_{\mathit{rel}} \leftarrow \{ \mathit{c} \} \bigcup C_n$.
    \STATE $C \leftarrow$ Set of all IDs of committees that the directory node is responsible for.
    \ELSE
    \STATE $C_{\mathit{rel}} \leftarrow \{ \mathit{c} \}$.
    \STATE $C \leftarrow$ Singleton set of ID of the committee that the node belongs to.
    \ENDIF
    \IF{$P_{\mathit{join}} < T_{\mathit{join}}$ \AND $C_{\mathit{rel}} \bigcap C$ is not $\emptyset$ \AND $\mathit{valid\_blk}$}
    \STATE $V \leftarrow V \bigcup \{(m, e)\}$.
    \ENDIF
\ENDFOR
\STATE Return $V$.
\end{algorithmic}
\end{algorithm}

\begin{algorithm}
\caption{\texttt{STORE\_INFO} protocol}
\label{alg:storeinfo}
\begin{algorithmic}
\REQUIRE Set $V$ consisting of all the valid \texttt{JOINING} messages in this round. Let $C_i$ be set of all the tuples of \texttt{JOINING} message and network address of valid nodes (that this directory node is aware of) in committee $i$. Let $\mathbf{H}(.)$ denote the hash function. Let $\mathbf{Blk}(i)$ be a function that returns the block from the confirmed chain with block number $i$ if it exists, otherwise, returns $\varnothing$.
\ENSURE Store entry information of all the nodes that provided a valid proof in this round.
\FOR{each $(\mathtt{JOINING}, e) \in V$}
    \STATE $\mathit{blk\_num}, N_c, \mathit{net\_addr}  \leftarrow e$. \COMMENT{Retrieve entry information.}
    \STATE $\mathit{blk} \leftarrow \mathbf{Blk}(\mathit{blk\_num})$.
    \STATE $P_{\mathit{join}} \leftarrow \mathbf{H}(\mathbf{H}(\mathit{blk})  \mathbin\Vert \mathit{net\_addr} \mathbin\Vert N_c)$.
    \STATE $\mathit{c} \leftarrow$ Leftmost $\lceil \log \mathcal{C}_e \rceil$ bits of $P_{\mathit{join}}$.
    \STATE $C_c \leftarrow C_c \bigcup \{e\}$.
\ENDFOR
\end{algorithmic}
\end{algorithm}

\begin{algorithm}
\caption{\texttt{REPLY\_INFO} protocol}
\label{alg:replyinfo}
\begin{algorithmic}
\REQUIRE Set $V$ consisting of all the valid \texttt{REQ\_INFO} messages in this round. Let $C_i$ be set of all the tuples of \texttt{REQ\_INFO} message and network address of valid nodes (that this directory node is aware of) in committee $i$.
\ENSURE Send committee entry information to nodes that provided a valid proof in this round.
\FOR{each $(\mathtt{REQ\_INFO}, c, e) \in V$}
    \STATE SEND (\texttt{COMM\_INFO}, $C_c$) to the node with network address $\mathit{net\_addr}$.
\ENDFOR
\end{algorithmic}
\end{algorithm}

\end{document}